\documentclass[a4paper,oneside,reqno]{amsart}
\usepackage{amssymb,amsfonts}

\usepackage{amssymb,amsfonts,latexsym}
\usepackage{amscd} 
\usepackage{mathrsfs}
\usepackage{graphicx}
\usepackage{xcolor}
\usepackage{float}
\usepackage{tikz}
\usetikzlibrary{arrows,shapes,positioning,automata,decorations.markings}
\tikzstyle arrowstyle=[scale=1]
\tikzstyle directed=[postaction={decorate,decoration={markings,mark=at position .5 with {\arrow[arrowstyle]{stealth}}}}]
\usepackage{multicol}
\usepackage[a4paper]{geometry}
\usepackage[only,llbracket,rrbracket]{stmaryrd}
\SetSymbolFont{stmry}{bold}{U}{stmry}{m}{n}
\usepackage[verbose]{hyperref}
\usepackage[capitalise,noabbrev,sort]{cleveref}
\hypersetup{colorlinks=false,allbordercolors=blue,pdfborderstyle={/S/U/W 1}}
\usepackage{url}
\usepackage{enumerate}

\newtheorem{theorem}{Theorem}[section]
\newtheorem{lemma}[theorem]{Lemma}
\newtheorem{proposition}[theorem]{Proposition}
\newtheorem{corollary}{Corollary}[theorem]

\theoremstyle{definition}
\newtheorem{definition}[theorem]{Definition}

\newtheorem{example}[theorem]{Example}
\newtheorem{remark}[theorem]{Remark}

\newcommand{\arabiclist}{enumerate}

\newcommand{\Dendron}{Final tree}
\newcommand{\dendron}{final tree}%+
\newcommand{\MCF}{multiple context-free}%+
\newcommand{\dom}{\operatorname{dom}}%+
\newcommand{\TSG}{\operatorname{TS}(\CCC)}%+
\newcommand{\up}{\operatorname{up}}%+
\newcommand{\down}{\operatorname{down}}%+
\newcommand{\true}{\texttt{true}}%+
\newcommand{\id}{\operatorname{id}}%+
\newcommand{\equals}{\operatorname{eq}}%+
\newcommand{\set}{\operatorname{set}}%+
\newcommand{\push}{\operatorname{push}}%+
\newcommand{\CCC}{C}%+
\newcommand{\Updownvector}{Up-down vector}
\newcommand{\updownvector}{up-down vector}%+
\newcommand{\historia}{history}%+
\newcommand{\Historyvec}{History array}
\newcommand{\historyvec}{history array}%+
\newcommand{\harry}{\mathbf{h}}%+
\newcommand{\candq}{(\mathbf c,\mathbf q)}%+
\newcommand{\UpSet}{up-set}%+
\newcommand{\Good}{Standardised}\newcommand{\good}{standardised}%+
\newcommand{\degree}{\deg}%+

\newcommand{\cA}{\mathcal A}%+
\newcommand{\cB}{\mathcal B}%+
\newcommand{\cM}{\mathcal M}%+
\newcommand{\cR}{\mathcal R}%+
\newcommand{\cU}{\mathcal U}%+ 
\newcommand{\s}{\sigma }%+
\renewcommand{\t}{\tau }%+
\newcommand{\GG}{\mathbb G}%+
\newcommand{\N}{\mathbb N}%+
\newcommand{\Z}{\mathbb Z}%+
\renewcommand{\geq}{\geqslant} 
\renewcommand{\leq}{\leqslant} 
\renewcommand{\ge}{\geqslant} 
\renewcommand{\le}{\leqslant}

\begin{document}

\title[A substitution lemma for multiple context-free languages]{A substitution lemma for  multiple context-free languages}

\author[A. Duncan]{Andrew Duncan}
\address{School of Mathematics, Statistics and Physics, Newcastle University,
Newcastle upon Tyne
NE1 7RU, United Kingdom}
\email{andrew.duncan@newcastle.ac.uk }

\author[M. Elder]{Murray Elder}
\address{School of Mathematical and Physical Sciences, University of Technology Sydney, Ultimo NSW 2007, Australia}
\email{murray.elder@uts.edu.au}

\author[L. Frenkel]{Lisa Frenkel}
\address{Independent Researcher, Tallinn, 10159, Estonia}
\email{lizzy.frenkel@gmail.com}

\author[M. Lyu]{Mengfan Lyu}
\address{School of Computer, Data and Mathematical Sciences, Western Sydney University, NSW, Australia}
\email{mengfan.lyu@westernsydney.edu.au}

\date{\today}



\maketitle

\begin{abstract}
We present a  necessary condition for an infinite  language to be multiple context-free, 
which we call a Substitution Lemma. We apply it  to show a sample selection of languages are not multiple context-free, including the word problem of the group $F_2\times F_2$. We also show  that groups with multiple context-free word problem  have decidable rational subset membership problem. 

Our result contrasts with previous work showing that the standard pumping lemma for context-free languages cannot be generalised to multiple context-free languages, and that weak variants of generalised Ogden's lemma do not apply to multiple context-free languages.

\end{abstract}

\section{Introduction}

\setcounter{footnote}{1}
Multiple context-free languages were introduced by Seki,  Matsumura, Fujii and Kasami \cite{Kasami2, Seki}. They were proposed as an attempt to describe the syntax of natural languages (see for example \cite{Clark}), but were later shown not to suffice when Salvati proved that the language MIX$_3$\footnote{MIX$_n=\{w\in \{a_1,\dots, a_n\}^*\mid |w|_{a_1}=\cdots= |w|_{a_n}\}$ where $ |w|_{x}$ is the number of occurrences of the letter $x$ in $w$.} is multiple context-free  \cite{Salvati}, since MIX$_3$ is regarded as a language that should not appear in any class capturing the syntax of natural languages \cite{Joshi_1985}.

A useful tool for working with formal language classes is to find a necessary condition for a language to belong to the class.
For example, the well-known pumping lemmas for regular and context-free languages (see for example \cite{sipser13}),  \emph{Ogden's lemma} (see Lemma~\ref{lem:Ogden}) and the \emph{Interchange lemma} for context-free languages \cite{Ogden, Interchange}, Gilman's {\em Shrinking lemma}  for indexed languages \cite{Gilman-Shrinking}, Mitrana and Stiebe's \emph{Interchange lemma}  and
the {\em Swapping lemma} of the second author   for blind multi-counter languages  \cite{Mitrana,E-BScounter}.

Seki \emph{et al.} \cite{Seki} gave what was later called a \emph{weak} pumping lemma for \MCF\ languages, which is weak in the sense that pumping only applies to a single word in the language, rather than all words that are sufficiently long (see Lemma~\ref{lem:Seki}). 
Kanazawa \emph{et al.} \cite{failure} proved that it was not possible to prove a certain version of
the standard pumping lemma for context-free languages to \MCF\ languages, and Kanazawa also showed  a variant of generalised Ogden's lemma is not possible \cite{KanazawaOgden}. Meanwhile, attempts have also been made to prove stronger versions of both lemmas for some subclasses of multiple context-free languages (for example, for well-nested multiple context-free languages there is a pumping lemma \cite{KanazawaWN}, strengthened and accompanied with Ogden's lemma in \cite{Sorokin}).

In this article we introduce a necessary condition for an infinite  language to be multiple context-free, which we call a {\em Substitution Lemma} (Theorem~\ref{thm:Substitution}),  as an alternative type of pumping lemma for \MCF\ languages, which is \emph{strong} in the sense that the substitution applies to all words in the language that are of sufficient length. 
Our result can be viewed as a  generalisation of Ogden's and the Interchange lemmas \cite{Ogden, Interchange}.
Its proof relies on a {characterisation} of $k$-\MCF\ languages as those accepted by \emph{$k$-restricted tree stack automata} due to Denkinger \cite{Denkinger}.
As a demonstration, in Section~\ref{sec:Applications} we use the Substitution Lemma to show that some  languages are not $k$-\MCF\ for some or all $k$, including languages whose Parikh image is semi-linear.

A key motivation for devising a new type of pumping lemma is to understand groups which have  \MCF\ word problem, see subsection~\ref{subsec:F2F2}. In Section~\ref{sec:RationalSubset} we observe that the rational subset intersection and membership problems are decidable for groups with \MCF\ word problem. 
Other recent work on properties and applications of  \MCF\ languages  includes \cite{PermClosure,LLsaw,cyclicShift}.

\section{Preliminaries}

We use $\N_0,\N_+$ to denote the sets of natural numbers starting at $0,1$ respectively. If $i,j\in \Z$ and $i\leq j$ we write $[i,j]$ for the set of integers $\{i,i+1,\dots, j-1,j\}$. 
If $\Sigma$ is a set, we let $\Sigma^*$ denote the set of all words (finite length strings) with letters from $\Sigma$, including the empty string $\varepsilon$ of length $0$, and $\Sigma^+$ the set of positive length words with letters from $\Sigma$.  If  $w\in\Sigma^*$ and $a\in\Sigma$ we  let $|w|$ denote the length of $w$, and
$|w|_{a}$ the number of occurrences of the letter $a$ in $w$.  We use $\sqcup$ to denote the disjoint union of two sets.
For any set $\Sigma$ and letter $x\not\in\Sigma$ we let $\Sigma_x$ denote $\Sigma\sqcup\{x\}$.
An \emph{alphabet} is a finite set. 
 If $u,v\in \Sigma^*$ we say $v$ is a \emph{factor} of $u$ if there exist $\alpha,\beta\in\Sigma^*$ so that $u=\alpha v\beta$.

\subsection{Multiple context-free grammars} 

See \cite{KS,Salvati} for more precise 
definitions; for the purposes of this paper we rely on an equivalent definition 
in terms of tree stack automata given in Subsection~\ref{subsec:TSA} below, so the grammar definition we present here 
 is somewhat informal.

Fix an alphabet  $N$ of \emph{non-terminals} and an alphabet $\Sigma$ of \emph{terminals}.
In a context-free grammar $G_0=(N,\Sigma,R,S)$, rules  
 in $R$ are of the form
\[A\to w_1B_1w_2B_2\dots w_nB_nw_{n+1}\]
where $A,B_1, \dots, B_n \in N$, $w_1, \dots, w_{n+1} \in \Sigma^*$, and $S\in N$ is a distinguished {\em start} non-terminal.
We could think of rules in another way:
 \[A(s_1)\leftarrow B_1({x_1}), B_2({x_2}),\dots, B_n({x_n})\] 
where $x_1,x_2,\dots,x_n$ are variables,  
$s_1=w_1{x_1}w_2{x_2}\dots  w_n{x_n}w_{n+1}$, $w_1, \dots, w_{n+1} \in \Sigma^*$, 
and non-terminals are  regarded as functions of a single variable, with the rule taking what has already been produced by each $B_i$, $i \in [1,n]$ (a word $u_i\in\Sigma^*$ in place of the variable $x_i$, which we can call the \emph{value} of $B_i$) and inserting these   into the string $s_1$, with the resulting word  becoming the value of the variable for $A$. Also, rules of the form $A\rightarrow t$, where $t\in \Sigma^*$, take the form
\[A(t)\leftarrow\]
and are used to initialise non-terminals. 
The language $L\subseteq \Sigma^*$ of a context-free grammar defined in this way
is the set of all strings $w\in \Sigma^*$ that can be obtained by a sequence of rules finishing with the non-terminal $S(w)$
(this is tantamount to working
from the leaves of the derivation tree to the root, as opposed the usual process of beginning with the start non-terminal and working towards the leaves.)

In a multiple context-free grammar $G=(N,\Sigma,R,S)$,  each non-terminal in $N$ is allowed to  have more than one variable, rules are of the form:
\[A(s_1,\dots,s_m)\leftarrow B_1({x_1^1},\dots, {x_{k_1}^1}), B_2({x_1^2},\dots, {x_{k_2}^2}),\dots, B_n({x_1^n},\dots, {x_{k_n}^n})\] and 
\[A(t_1,\dots,t_m)\leftarrow\] 
where  $s_1,\dots, s_m$ are strings of terminals and 
variables $x_j^i$, where the $x_j^i$ are distinct and can only be used at most once in $s_1\cdots s_m$ (the concatenation of all the $s_i$),    $t_i$ are strings of terminals only (so that the grammar is \emph{initialised} with words in place of variables by rules of the form  $A(t_1,\dots,t_m)\leftarrow$). A non-terminal $N(x_1,\dots,x_m)$ with $m$ variables is said to have  \emph{rank} $m$. Here $S\in N$ is a distinguished {\em start} non-terminal which has rank 1 (and is
applied  last). 

We call a multiple context-free grammar {\em $k$-multiple context-free}, abbreviated as $k$-MCF, if each non-terminal has rank at most $k$. By this construction,  any context-free  grammar  is $1$-MCF. 
The language $L\subseteq \Sigma^*$ of a multiple context-free grammar is the set of all strings $w\in \Sigma^*$ that can be obtained by a sequence of rules finishing with the non-terminal $S(w)$.
A language $L$ of a $k$-multiple context-free grammar $G=(N,\Sigma,R,S)$ is called a $k$-multiple context-free language, and $L\subseteq \Sigma^*$ is called multiple context-free if it is $k$-\MCF\ for some integer $k\geq 1$.

\begin{example}\label{abcd}
The language  $\{a^mb^mc^md^m\mid m\in \N_0\}$ is  2-multiple context-free, produced by the grammar
\[\begin{array}{rcl}
T(\varepsilon,  \varepsilon)& \leftarrow &\\
T(a{x_1}b,c{x_2}d)& \leftarrow &T({x_1},{x_2})\\
S({x_1}{x_2})& \leftarrow &T({x_1},{x_2})\\
\end{array}\]
\end{example}

\begin{example}\label{abcdnm}
The language  $\{a^nb^mc^nd^m\mid m,n\in \N_0\}$ is 2-multiple context-free, produced by the grammar
\[\begin{array}{rclcrcl}
P(\varepsilon, \varepsilon)& \leftarrow &  & Q(\varepsilon, \varepsilon)& \leftarrow &\\
P(a{x_1},c{x_2}) & \leftarrow &P({x_1},{x_2})& Q(b{x_1},d{x_2}) & \leftarrow &Q({x_1},{x_2})\\
S({x_1}{y_1}{x_2}{y_2})& \leftarrow &P({x_1},{x_2}), Q({y_1},{y_2})\\
\end{array}\]
\end{example}

\begin{proposition}[See for example \cite{Seki}]\label{prop:closure_props}
If $\Sigma,\Gamma$ are alphabets, $L,L_1\subseteq \Sigma^*$  are $k$-multiple context-free and $R\subseteq \Sigma^*$ is a regular language then 
\begin{\arabiclist}
\item\label{item:1}  $L\cup L_1$ is $k$-multiple context-free;
\item\label{item:2} $L\cap R$ is $k$-multiple context-free;
\item\label{item:3}  $\psi(L), \tau^{-1}(L)$ are $k$-multiple context-free for any homomorphisms $\psi\colon\Sigma^*\to \Gamma^*, \tau\colon\Gamma^*\to \Sigma^*$.
\end{\arabiclist}\end{proposition}

\subsection{Semi-linear sets}\label{subsec:semi-linear}

An important concept in formal language theory is that of semi-linear sets. See for example  \cite{Bondarenko,Brough,Ginsburg} for more details. 
For  $d\in\N_+$ a set 
$A\subseteq \N_0^d$ is called \emph{linear} if there exist elements (vectors) $v_0,v_1,\dots, v_m\in\N_0^d$ so that \[A=\{v_0+\alpha_1v_1+\dots +\alpha_mv_m\mid \alpha_1, \dots, \alpha_m \in \N_0\},\] and $A$ is called \emph{semi-linear} if it is equal to a finite union of linear subsets of $ \N_0^d$.
If $\Sigma=\{a_1,\dots, a_d\}$ then  the map $\psi\colon \Sigma^*\to \N_0^d$ defined by $\psi(w)=(|w|_{a_1}, \dots, |w|_{a_d})$ is called the \emph{Parikh map}. Parikh showed that  if $L$ is context-free then $\psi(L)$ is semi-linear.
Vijay-Shanker,  Weir,  and Joshi \cite{Vijay} showed that  if $L$ is  \MCF\ then $\psi(L)$ is semi-linear. We list in the next lemma further facts about semi-linear sets.
\begin{lemma}[See for example \cite{Ginsburg}] \label{lem:SemiL-properties}
The class of semi-linear sets {is} closed under union, intersection, and complement. Moreover,
\begin{\arabiclist}
 \item if  
 $L\subseteq \Sigma^*$ and $\psi(L)$ is semi-linear then there exists a regular language $R\subseteq \Sigma^*$ so that $\psi(L)=\psi(R)$;
 \item if $L$ has semi-linear Parikh image and $R$ is a regular language then  $L\cap R$ has semi-linear Parikh image.
\end{\arabiclist}
\end{lemma}

We also observe the following (c.f.  \cite[Proposition 3.2]{Brough}). 

\begin{lemma}[Gap lemma for preimages of semi-linear sets]
\label{lem:GapSemi} If $L\subseteq \Sigma^*$ is an infinite language
such that for every $m\in\N$ there exists $N_m\in\N$ such that $|w_1|-|w_2|>m$ for all $w_1,w_2\in L$ with $|w_1|>|w_2|>N_m$,
then the Parikh image of $L$ is not semi-linear.\end{lemma}
\begin{proof} 
By Lemma~\ref{lem:SemiL-properties}
there exists a regular language $R$ with $\psi(R)=\psi(L)$.
Fix a deterministic finite state automaton (DFA) accepting $R$ and let $p$ be the number of states of this DFA. 
 Take $m>p$, and let  $N_m$  be as in the hypothesis. 
Since $L$ is infinite, $R$ is infinite so we can find a word $u_1\in R$ with $|u_1|>N_m+m+p$, and by the pumping lemma for regular languages $u_1=xyz$ where $|y|\in[1,p]$
and $u_2=xz\in R$. However, $\psi(u_1)$ and $\psi(u_2)$ have preimages $w_1,w_2$ in $L$ with $|u_i|=|w_i|$ for $i=1,2$. Thus $|w_1|>|w_2|>N_m+m$, but  $|w_1|-|w_2|\leq p$, which is a contradiction.
\end{proof} 

It follows that languages such as \begin{multicols}{2}\begin{itemize}
\item	$L_1=\{a^{2^n} \mid n\in\N\}$,
\item	$L_2=\{a^{n^2}\mid n\in\N\}$,\columnbreak
\item	$L_3=\{a^{ \lfloor n\log n \rfloor}\mid n\in\N\}$,
\item	$L_4(\alpha)=\{a^{\lfloor n^\alpha \rfloor}\mid n\in\N\}, \alpha>1$
\end{itemize}\end{multicols}
do not have semi-linear Parikh image, and therefore are not \MCF. 
More generally we have the following
\begin{corollary}
  Let $f:\N\rightarrow\N$ be a function such that for any $m\in \N$, there exists $n\in\N$ so that $|f(n_1)-f(n_2)|>m$ for all $n_1>n_2>n$.
  Then the language $L:=\{a^{f(n)}\colon n\in\N\}$
does not have semi-linear Parikh image, and  thus is not \MCF.\end{corollary}

\subsection{Tree stack automata}\label{subsec:TSA}

Denkinger \cite{Denkinger} showed that corresponding to multiple context-free grammars there is a machine called a \emph{tree stack automaton} accepting the same language. To explain this we first recall the definitions of trees and tree stacks as in \cite{Denkinger}.

\begin{definition}[Tree with labels in $\CCC$]
\label{defn:TreeWithLabels}
Let  $\CCC$ be an alphabet and $@\not\in\CCC$.
  Let $\xi$ be a  partial function from $\N_+^*$ to $\CCC_@$ ($=C\sqcup\{@\}$) with a non-empty, finite and prefix-closed domain denoted $\dom(\xi)$ such that $\xi(\nu)=@$ if and only if $\nu=\varepsilon$. 
  Then $\xi$ defines a labelled  rooted tree $T$ with vertices $V(T)=\dom(\xi)$ where $\nu\in V(T)$ is labelled $\xi(\nu)$,  root $\varepsilon$, and edges $E(T)=\{\{\nu,\nu n\} \mid \nu n \in \dom(\xi)\}$.
\end{definition}

Note that $\dom(\xi)$ being non-empty and prefix-closed ensures that every vertex is connected to the root vertex by an edge path. We will occasionally abuse notation and write $\xi=T$.

\begin{remark} 
Note that there is no requirement in the definition of a tree that a vertex  $\rho$ with a child $\rho n$, $n>1$, also has a child $\rho (n-1)$. 
  \end{remark}

\begin{definition}[Tree stack over $\CCC$]
A \emph{tree stack} 
 is a pair $(\xi,\rho)$, where 
 $\rho\in \dom(\xi)$. 
The set of all tree stacks over $\CCC$ is denoted by $\TSG$.
\end{definition}
We refer to the address $\rho$ as the \emph{pointer}, and visualise the tree stack as a rooted tree with root at address $\varepsilon$ labeled $@$ and all other vertices labeled by letters from $\CCC$, with the \emph{pointer} $\rho$ indicated by an arrow pointing to  the vertex at address $\rho$.
See for example Fig.~\ref{fig:TreeStackEG}.

We note that Denkinger's notation treats the root vertex slightly differently, but for ease of exposition we have chosen the conventions above. 

\begin{figure}[!htb]
\begin{center}
\begin{tikzpicture}[scale=1]
\draw[decorate]  (0,2) --  (1,1) -- (2,0) -- (3,1);
\draw[decorate] (1,1) -- (2,2);

\draw (0,2) node {$\bullet$};
\draw (1,1) node {$\bullet$};
\draw (2,2) node {$\bullet$};
\draw (2,0) node {$\bullet$};
\draw (3,1) node {$\bullet$};

\draw (2.1,-0.3) node {$@$};
\draw (0.1,2.3) node {$\#$};
\draw (2.1,2.3) node {$\dag$};

\draw (3.1,1.3) node {$\dag$};
\draw (1,1.3) node {$\ast$};

\draw[decorate,ultra thick,->] (-1,2) -- (-0.4,2);
\draw (-0.2,1.85) node {\color{cyan}\tiny$14$};
\draw (2.15,1.85) node {\color{cyan}\tiny$16$};
\draw (0.85,0.85) node {\color{cyan}\tiny$1$};
\draw (3.15,0.85) node {\color{cyan}\tiny$3$};
\draw (1.85,-0.15) node {\color{cyan}\tiny$\varepsilon$};
\end{tikzpicture}

\end{center}

\caption{Tree stack $(\xi, 14)$  with $\xi$  defined by  $\xi\colon \varepsilon\mapsto @, 1 \mapsto \ast, 14\mapsto \#, 16\mapsto \dag, 3\mapsto \dag$ and pointer at $14$. Here $\dom(\xi)=\{\varepsilon, 1,14, 16, 3\}$ and $\{\ast,  \#, \dag\}\subseteq C$.\label{fig:TreeStackEG}}
\end{figure}

To define a tree stack automaton, we employ the following predicates and {(partial)} functions. 
The \emph{predicates} are used to check if the pointer of a tree stack has a certain label, namely: 
\begin{\arabiclist} 
\item $\equals(c)=\{(\xi, \rho) \in \TSG \mid \xi(\rho)=c\}$ where $c \in \CCC_@$ (so $(\xi, \rho)\in  \equals(c)$, or $\equals(c)$ is true for the tree stack  $(\xi, \rho)$, if the label of $\rho$ is $c$)\footnote{In \cite{Denkinger} Denkinger calls the predicate $\equals(@)$ ``bottom'' and treats it separately.},
\item $\true$ coincides with $\TSG$ (true for any tree stack regardless of the current label pointed to).
\end{\arabiclist}

We list below the (partial) \emph{functions} that provide both construction of and movement around a tree stack: 
  \begin{\arabiclist} 
  \item $\id\colon \TSG \rightarrow \TSG$ where $\id(\xi, \rho)=(\xi, \rho)$ for every $(\xi, \rho) \in \TSG$. This is a function that makes no change to the tree stack.
\item  For each $(n,c)\in \N_+\times \CCC$ let  $\push_n(c) \colon \TSG  \rightarrow \TSG$  be the
  partial function defined,  whenever $\rho n$ is not an address in $\dom(\xi)$, 
  by $\push_n(c)((\xi, \rho))=(\xi', \rho n)$, 
  where $\xi'$ is the partial function defined as 
\[\xi'(\nu)=\begin{cases}\xi(\nu), & \nu\in\dom(\xi)\\
  c, & \nu=\rho n\end{cases}.\]
Thus $\push_n(c)$  adds a new vertex $\rho n$ labelled by $c$, a new edge from $\rho$ to $\rho n$ and moves the pointer from $\rho$ to $\rho n$,  provided $\rho n$ was not already a vertex of $\xi$.
\item   For each $n\in \N_+$ let $\up_n \colon \TSG \rightarrow \TSG$ be 
  the partial function defined,  whenever  $\rho n \in \dom(\xi)$,  by $\up_n((\xi, \rho))=(\xi, \rho n)$.
  So $\up_n$ moves the pointer of a tree stack to an existing child address.
\item   $\down \colon \TSG \rightarrow \TSG$ is a partial function defined as $\down(\xi, \rho n)=(\xi, \rho)$ (so down  moves the pointer from a vertex to its parent and may be  applied to any tree stack unless its
  pointer is already pointing to the root, which has no parent).
\item  For each $c\in \CCC$ let $\set(c) \colon \TSG  \rightarrow \TSG$  be the function defined,
  whenever $\rho\neq \varepsilon$, by  $\set(c)( (\xi,\rho))=(\xi',\rho)$, where $\xi'$ is the partial function  defined as 
\[\xi'(\nu)=\begin{cases}\xi(\nu), & \nu\in\dom(\xi)\setminus\{\rho\}\\
    c, & \nu=\rho\end{cases}.\]
Thus $\set(c)$ relabels the vertex pointed to by the pointer with the letter $c\in \CCC$ (and may be  applied to any tree stack except $(\xi,\varepsilon)$). 
\end{\arabiclist}

Note that none of the partial functions $\id$, $\push_n(c)$, $\up_n$, $\down$ or $\set(c)$ remove any vertices from a tree stack, and only $\push_n(c)$ adds a vertex. 

\begin{definition}[Tree stack automata]
A \emph{tree stack automaton (TSA)} is a tuple \[\mathcal{A}=(Q,\CCC, \Sigma, q_0,  \delta, Q_f)\] where $Q$ is a finite  \emph{state set}, $\CCC_@$ is an alphabet of \emph{tree-labels} (with $@ \notin \CCC$), $\Sigma$ is the alphabet of \emph{terminals}, $q_0\in Q$ is the \emph{initial state},  
 $Q_f\subseteq Q$ is the subset of \emph{final states}, and $\delta = \{\s_1,\dots, \s_{|\delta|}\}$
  is a finite set of \emph{transition rules} of  the form $\s_i=(q,x,p,f,q')$ where 
 \begin{\arabiclist}
 \item  $q,q'\in Q$ are the \emph{source} and \emph{target} states of $\s_i$,
 \item  $x\in\Sigma_\varepsilon$ is the \emph{input letter} for $\s_i$, 
 \item $p$ is the predicate of $\s_i$ (i.e. either $\equals(c)$, for some $c\in\CCC_@$, or $\true$),
 \item $f$ is the \emph{instruction} of $\s_i$, so is one of the functions $\id$, $\push_n(c)$, $\up_n$, $\down$, or $\set(c)$ for $c\in\CCC, n\in \N_+$.
 \end{\arabiclist}
\end{definition} 
Note that each coordinate of $\mathcal{A}$ has a finite description.
The tree stack automaton operates by starting in state $q_0$ with a tree stack initialised at $(\{(\varepsilon,@)\},\varepsilon)$,
reading each letter of an input word $w\in\Sigma^*$ with arbitrarily many $\varepsilon$ letters interspersed, applying transition rules from $\delta$, if they are allowed, (eg. $\push_n$ unless $\rho n$ is already a vertex;  the target state of the applied transition should also coincide with the source state of the following one). In other words, a composition of transitions is allowed if it determines
a well defined function which consumes input letters whilst moving between tree stacks. 
An allowed  sequence of transitions  $\tau_1\cdots \tau_r\in\delta^*$ is called a \emph{run} of the automaton. 
A run is \emph{valid} if it finishes in a state in $Q_f$; if $w\in \Sigma^*$ is the input word consumed during a valid run we say that $\mathcal A$ \emph{accepts} 
$w$ (we do not necessarily put a requirement on the final position of the pointer, but see Remark \ref{rmk:finish-root}).
Define the language $L(\mathcal A)$ to be the set of all words $w\in\Sigma^*$ accepted by $\mathcal A$.

\begin{example}\label{eg:ABCD}
Let $Q=\{q_0,q_1,q_2,q_3,q_4\}, \Sigma=\{a,b,c,d\}$ and $\CCC=\{*, \#\}$. Consider the TSA
\[
\mathcal{A}=(Q, \CCC, \Sigma, q_0, \delta, \{q_4\})
\]
where $\delta$ consists of the transitions
\begin{equation*}
\begin{array}{llllllllll}
 & \s_1=(q_{0}, {a}, \true, & \push_1(*), q_{0}),  & \hspace{1cm} &\s_6=(q_{2}, c, \equals(*), &\up_1, q_{2}), \\
 &\s_2=(q_{0}, \varepsilon, \true, & \push_1(\#), q_{1}), & &\s_7=(q_{2}, \varepsilon, \equals(\#),& \down, q_{3}),\\
 &\s_3=(q_{1}, \varepsilon, \equals(\#) ,& \down, q_{1}),  & &\s_8=(q_{3}, d, \equals(*),& \down, q_{3}),  \\
 &\s_4=(q_{1}, b,\equals(*), &\down, q_1), &  &\s_9=(q_{3}, \varepsilon, \equals(@),& \id, q_{4})\\
&\s_5=(q_{1},  \varepsilon, \equals(@), &\up_1, q_{2}), 
\end{array}
\end{equation*}
depicted by the labelled directed graph in Fig.~\ref{fig:EGabcdAUTOM}.

\begin{figure}[H]
\centering
\begin{tikzpicture}[shorten >=1pt,node distance=3cm,on grid,auto]

 \node[state,initial, initial where=below]  (1)                      {$q_0$};
  \node[state]          (2) [right =of 1] {$q_1$};
  \node[state]          (3) [right=of 2] {$q_2$};
    \node[state]          (4) [right=of 3] {$q_3$};
  \node[state,accepting](5) [right=of 4] {$q_4$};

  \path[->] (1) edge              node        { $\begin{array}{c}\varepsilon, \true ,\\ \push_1(\#)\end{array}$} (2)
                    edge [loop above] node        {$\begin{array}{c}{a},\true , \\\push_1(*)\end{array}$} ()
            (2) edge              node        {$\begin{array}{c}\varepsilon, \equals(@), \\\up_1\end{array}$} (3)
                       edge [loop below] node        {$\begin{array}{c}\varepsilon,\equals(\#) , \\\down\end{array}$} ()
                           edge [loop above] node        {$\begin{array}{c}b,\equals(*) ,\\\down\end{array}$} ()
            (3) edge              node {$ \begin{array}{c}\varepsilon, \equals(\#), \\\down\end{array}$} (4)
                  edge [loop below] node        {$\begin{array}{c}c,\equals(*), \\\up_1\end{array}$} ()
                              (4) edge              node {$\begin{array}{c} \varepsilon, \equals(@), \\\id\end{array}  $} (5)
                  edge [loop below] node        {$\begin{array}{c}d,\equals(*), \\\down\end{array}$} ()
                    ;
                    
\end{tikzpicture}
\caption{Finite state control encoding $\mathcal{A}$ accepting $a^mb^mc^md^m$, $m \geq 0$ in Example~\ref{eg:ABCD}. 
Each directed edge corresponds to a transition $\s_i\in\delta$.
\label{fig:EGabcdAUTOM}}

\end{figure}

\begin{figure}[H]
\centering

\begin{tikzpicture}[scale=1]
\draw[decorate] (0,3.9)-- (0,2.6) --  (0,1.3) -- (0,0);
\draw (0,3.9) node {$\bullet$};
\draw (0,2.6) node {$\bullet$};
\draw (0,1.3) node {$\bullet$};
\draw (0,0) node {$\bullet$};
\draw (0.3,0) node {$@$};
\draw (0.3,3.9) node {$\#$};
\draw (0.3,1.3) node {$\ast$};
\draw (0.3,2.6) node {$\ast$};
\draw[decorate,ultra thick,->] (-1,1.3) -- (-0.4,1.3);
\draw (-0.15,-0.15) node {\color{cyan}\tiny $\varepsilon$};
\draw (-0.15,1.25) node {\color{cyan}\tiny $1$};
\draw (-0.2,2.55) node {\color{cyan}\tiny $11$};
\draw (-0.25,3.85) node {\color{cyan}\tiny $111$};
\end{tikzpicture}
\caption{Tree stack after  having read $a^2b=aa\varepsilon\varepsilon b$ (on the way to accepting $a^2b^2c^2d^2$) using transitions $\s_1^2\s_2\s_3\s_4$ in Example~\ref{eg:ABCD}.
\label{fig:EGabcdTREE}}
\end{figure}

Then  $L(\mathcal A)=\{a^mb^mc^md^m|m\in \N_0\}$, as in Example~\ref{abcd}. 
For instance, 
the TSA $\mathcal{A}$ accepts $a^2b^2c^2d^2$ by using the sequence of transitions $\mathcal R= \s_1^2 \s_2 \s_3 \s_4^2 \s_5 \s_6^2 \s_7 \s_8^2 \s_9$ shown in Table~\ref{table:ABCD}.

\begin{table}
\begin{tabular}{|c|c|rl|c|} \hline%toprule
Transition & State  & Tree stack && Input read\\
\hline
& $q_{0}$ & $( \{  {(\varepsilon, @)} \},$ & $\varepsilon)$ & \\
$\vdash_{\s_1}$ & $q_0$ & $(\{(\varepsilon, @), {(1, *)}\},$ & $1)$ & $a$ \\
$\vdash_{\s_1}$ & $q_0$ & $(\{(\varepsilon, @), (1,*), {(11, *)}\},$ & $11)$ & $a$ \\
$\vdash_{\s_2}$ & $q_1$ & $(\{(\varepsilon, @), (1, *), (11,*), {(111, \#)}\},$ & $111)$ & $\varepsilon$ \\
$\vdash_{\s_3}$ & $q_1$ & $(\{(\varepsilon, @), (1,*), {(11, *)}, (111, \#)\},$ & $11)$ & $\varepsilon$ \\
$\vdash_{\s_4}$ & $q_1$ & $(\{(\varepsilon, @), (1, *), (11,*), {(111, \#)}\},$ & $1)$ & $b$ \\
$\vdash_{\s_4}$ & $q_1$ & $(\{(\varepsilon, @), (1, *), (11,*), {(111, \#)}\},$ & $\varepsilon)$ & $b$ \\
$\vdash_{\s_5}$ & $q_2$ & $(\{(\varepsilon, @), (1, *), (11,*), {(111, \#)}\},$ & $1)$ & $\varepsilon$ \\
$\vdash_{\s_6}$ & $q_2$ & $(\{(\varepsilon, @), (1, *), (11,*), {(111, \#)}\},$ & $11)$ & $c$ \\
$\vdash_{\s_6}$ & $q_2$ & $(\{(\varepsilon, @), (1, *), (11,*), {(111, \#)}\},$ & $111)$ & $c$ \\
$\vdash_{\s_7}$ & $q_3$ & $(\{(\varepsilon, @), (1, *), (11,*), {(111, \#)}\},$ & $11)$ & $\varepsilon$  \\
$\vdash_{\s_8}$ & $q_3$ & $(\{(\varepsilon, @), (1, *), (11,*), {(111, \#)}\},$ & $1)$ & $d$ \\
$\vdash_{\s_8}$ & $q_3$ & $(\{(\varepsilon, @), (1, *), (11,*), {(111, \#)}\},$ & $\varepsilon)$ & $d$ \\
$\vdash_{\s_9}$ & $q_4$ & $(\{(\varepsilon, @), (1, *), (11,*), {(111, \#)}\},$ & $\varepsilon)$ & $\varepsilon$ \\
\hline
\end{tabular}
\caption{Sequence of transitions to accept $a^2b^2c^2d^2$ in Example~\ref{eg:ABCD}.\label{table:ABCD}}
\end{table}

The tree stack after having read $a^2b$ using transitions $\s_1^2\s_2 \s_3 \s_4$ is depicted in Fig.~\ref{fig:EGabcdTREE}.
\end{example}

\begin{remark}[Tree-label alphabet is non-empty]\label{rmk:node-label-non-empty}
Without loss of generality we may assume the tree-label alphabet $C$ is  non-empty for any TSA.
That is, if a TSA is given where $C=\emptyset$ (in which case the tree-stack is  $(\varepsilon,@)$ throughout every run and the language of the TSA is regular),
we can modify its  description by setting $C=\{\#\}$,  adding a new start state $q_0'$ and two transitions  \[(q_{0}', {\varepsilon}, \equals(@), \push_1(\#), q_{0}'),   \hspace{1cm} (q_{0}', \varepsilon, \equals(\#), \down, q_{0})\]
so that every accepting run must first follow these transitions. The modified TSA accepts the same language as the original so, as claimed, we may always assume
$C\neq \emptyset$.    
\end{remark}

A non-empty run of the TSA $\mathcal A$ has the form $\tau_1 \cdots \tau_r$, for some transitions $\tau_j \in \delta, j \in [1,r].$
Denote by $(\xi_j,\rho_j)$ the tree stack after performing $\tau_1\cdots \tau_j$,  
and $(\xi_0,\rho_0) = (\{\varepsilon,@\},\varepsilon)$ the tree stack before any transitions are performed.

\begin{definition}[Degree]
\label{defn:degree}
If $\delta$ is the set of  transition rules for a TSA $\mathcal A$, let $\Delta=\{i\in \N_+ \mid \exists c\in C, \push_i(c)\in \delta\}$. This means that in any tree stack $(\xi,\rho)$ built  by $\mathcal A$ whilst reading any input, $\dom(\xi)\subseteq \Delta^*$, so the out-degree of  any vertex of $\xi$ is at most
$|\Delta|$.
 Define the 
 \emph{degree} of $\mathcal A$ to be $\degree(\mathcal A)=|\Delta|$.
\end{definition}
Since $\delta$ is finite, $\degree(\mathcal A)$ is a non-negative integer. 
Recall (Definition~\ref{defn:TreeWithLabels})
that we do not require  nodes $\nu n$ in the final tree of a run to have a sibling $\nu (n-1)$, so it is possible that $\Delta$ contains integers larger than  $\degree(\mathcal A)$. We make the following assumption.

\begin{remark}\label{rmk;WLOGconsecutiveChildren}
 Without loss of generality 
 we may assume \[\max\{i\in \N_+ \mid \exists c\in C, \push_i(c)\in \delta\}=\degree(\mathcal A),\] 
since if not, there exists $j\in \N_+$ so that $\mathcal A$ has an instruction $\push_{j}(c)$ for some $c\in\CCC$ but no instruction $\push_{j-1}(c')$ for any $c'\in\CCC$. 
Then replacing all rules $\push_j,\up_j$ in $\delta$ 
by $\push_{j-1},\up_{j-1}$ neither changes the language accepted nor affect the number of times a node is visited from below.
\end{remark}

\begin{definition}[Below/above] \label{defn:below-above}
In a rooted tree $T$ with vertex $\nu$, we say a vertex  is \emph{below} $\nu$ if it lies in the connected component of $T\setminus \{\nu\}$ containing the root, and \emph{above} $\nu$ if it is not below (so $\nu$ itself is considered to be \emph{above} $\nu$ by this definition). Equivalently, $\mu$ is below $\nu$ if $\nu$ is not a prefix of $\mu$, and above $\nu$  if $\nu$ is a prefix of $\mu$.\end{definition}

The next restriction puts a bound on the number of times the pointer can move a step in the direction away from the root (by a $\push_n$ or $\up_n$ instruction) at a given vertex.

\begin{definition}[Visited from below]
Let $\nu'\in \N_+^*$, $n\in\N_+$, $\nu=\nu' n$, and suppose that some run $\tau_1 \cdots \tau_r$  of a TSA  $\mathcal A$ produces a tree stack which includes the vertex with address $\nu$. We say that $\nu$ is 
\emph{visited from below} by $\tau_j$
if $\tau_j=(q,x,p,\push_n(c),q')$, $c\in\CCC$ or  $\tau_j=(q,x,p,\operatorname{up}_n,q')$ is applied when  the pointer is at $\nu'$.
\end{definition}
\begin{definition}[$k$-restricted]
 We say the run  $\tau_1 \cdots \tau_r$ is \emph{$k$-restricted} if for each $\nu\in \N_+^+$ the number of 
$j\in[1,r]$ such  that 
$\nu$ is visited from below by $\tau_j$ is at most $k$. We say that 
a TSA $\mathcal A$ is \emph{$k$-restricted} if for any word $w\in L(\mathcal A)$
there is an accepting run that is $k$-restricted. \end{definition}
 For example, one may verify that the automaton in Example~\ref{eg:ABCD} accepting the language $\{a^nb^nc^nd^n\mid n\in\N_0\}$ is 2-restricted as follows: each word $a^pb^pc^pd^p$ is accepted by a run $\s_1^p \s_2 \s_3 \s_4^p \s_5 \s_6^p \s_7\s_8^p \s_9$; each vertex address $1^i$ for $i\in[1,p+1]$  is visited from below once with a $\s_1$ or $\s_2$ ($\push_1$), and once with a $\s_5$ or $\s_6$ ($\up_1$).  See Table~\ref{table:ABCD} for the case $p=2$. Note that this language is also $2$-\MCF\ (Example~\ref{abcd}).  
 As another example, one may verify that every context-free language is accepted by a TSA that is 1-restricted (see \cite[Example 3.21]{KS}).
 
We now state the key result for  $k$-restricted tree stack automata due to  Denkinger.
\begin{theorem}[Theorem 4.12, \cite{Denkinger}]$\label{equ}$
Fix an alphabet $\Sigma$. Let $L \subseteq \Sigma^*$ and $k \in \N_{+}$. The following are equivalent:
\begin{\arabiclist}
\item 
 there is a $k$-multiple context-free grammar ${G}$ which generates  $L$;
\item 
there is a $k$-restricted tree stack automaton $\mathcal{A}$ which recognises  $L$.
\end{\arabiclist}
\end{theorem}

In the following sections, we use this characterisation  to prove our results.

\section{Preliminary results: \good, proper, \dendron, factorisation, \historia\ vector, \UpSet, single swap}\label{sec:prelims}

  \begin{remark}[Pointer returns to the root]\label{rmk:finish-root}
 Without loss of generality (by modifying the $k$-restricted TSA description if required) we may assume $\mathcal A$ accepts inputs only if it finishes in a state in $Q_f$ with the pointer pointing to the root. 
Such a modification does not affect the number of times a node is visited from below by a run, since we simply add additional $\down$ transitions.
 \end{remark}

\begin{definition}[\Good]\label{def:standardised} A tree stack automaton $\mathcal{A}=(Q,\CCC, \Sigma, q_0,  \delta, Q_f)$ is called \emph{\good} if the following holds:
if $ (q_1,\varepsilon,p_1,f_1,q_2), (q_2,\varepsilon,p_2,f_2,q_3) \in \delta$ and $f_1,f_2\in\{\id,\set(x)\mid x\in \CCC\}$  then $(q_1,\varepsilon,p_3,f_3,q_3) \in \delta$ for the values listed in Table~\ref{table:standardising}.

%\begin{table}[!h]
%\tbl{Values in Definition~\ref{def:standardised}. \label{table:standardising}}
%{\begin{tabular}{|l|l|} 
%\hline
% $p_1=p_2=\true$ & $f_3=f_2\circ f_1, p_3=p_2$\\
%\hline
%$p_1=\true, p_2=\equals(c), f_1=\id,  f_2=\id$ & $f_3=f_1, p_3=p_2$\\
%\hline
%$p_1=\true, p_2=\equals(c), f_1=\set(c),  f_2=\id$ & $f_3=f_1, p_3=p_1$\\
%\hline
%$p_1=\true, p_2=\equals(c), f_1=\id,   f_2\neq \id$ & $f_3=f_2, p_3=p_2$\\
%\hline
%$p_1=\true, p_2=\equals(c), f_1=\set(c),   f_2\neq \id$ & $f_3=f_2, p_3=p_1$\\
%\hline
%$p_1=\equals(c), p_2=\true, f_2=\id$ & $f_3=f_1, p_3=p_1$\\
%\hline
%$p_1=\equals(c), p_2=\equals(c),  f_1=\id\text{ or }\set(c), \forall f_2$  & $f_3=f_2, p_3=p_1$\\
%\hline
%$p_1=\equals(c), p_2=\equals(d),  c\neq d, f_1=\set(d),  f_2=\id$ & $f_3=f_1, p_3=p_1$\\
%\hline 
%$p_1=\equals(c), p_2=\equals(d),  c\neq d, f_1=\set(d),  f_2\neq\id$ & $f_3=f_2, p_3=p_1$\\
%\hline
%\end{tabular}}
%\end{table}

\begin{table}[!h]
\begin{tabular}{|l|l|} 
\hline
 $p_1=p_2=\true$ & $f_3=f_2\circ f_1, p_3=p_2$\\
\hline
$p_1=\true, p_2=\equals(c), f_1=\id,  f_2=\id$ & $f_3=f_1, p_3=p_2$\\
\hline
$p_1=\true, p_2=\equals(c), f_1=\set(c),  f_2=\id$ & $f_3=f_1, p_3=p_1$\\
\hline
$p_1=\true, p_2=\equals(c), f_1=\id,   f_2\neq \id$ & $f_3=f_2, p_3=p_2$\\
\hline
$p_1=\true, p_2=\equals(c), f_1=\set(c),   f_2\neq \id$ & $f_3=f_2, p_3=p_1$\\
\hline
$p_1=\equals(c), p_2=\true, f_2=\id$ & $f_3=f_1, p_3=p_1$\\
\hline
$p_1=\equals(c), p_2=\equals(c),  f_1=\id\text{ or }\set(c), \forall f_2$  & $f_3=f_2, p_3=p_1$\\
\hline
$p_1=\equals(c), p_2=\equals(d),  c\neq d, f_1=\set(d),  f_2=\id$ & $f_3=f_1, p_3=p_1$\\
\hline 
$p_1=\equals(c), p_2=\equals(d),  c\neq d, f_1=\set(d),  f_2\neq\id$ & $f_3=f_2, p_3=p_1$\\
\hline
\end{tabular}\caption{Values in Definition~\ref{def:standardised}. \label{table:standardising}}
\end{table}

\end{definition}

In other words, if there are two $\varepsilon$ transitions that do not change the position of the pointer, which can be applied one after the other, 
then there is a single transition which achieves the same outcome.

\begin{remark}\label{rmk:standardised}
If $\mathcal A$ is not \good, it is clear that we can add transitions $(q_1,\varepsilon,p_3,f_3,q_3)$ without changing the set of words accepted (reminiscent of the method of Benois \cite{ben69}), so without loss of generality we may assume all TSA are \good. 
See for instance Fig.~\ref{fig:making_\good} which illustrates adding a transition to satisfy the final row of Table~\ref{table:standardising}.

\begin{figure}[H]
\begin{center}
\begin{tikzpicture}[shorten >=1pt,node distance=3.1cm,on grid,auto]

  \node[state]  (1)                      {$q_1$};
  \node[state]          (2) [right =of 1] {$q_2$};
  \node[state]          (3) [ right=of 2] {$q_3$};

  \path[->] (1) edge              node        { $\begin{array}{c}\varepsilon, \equals(c) ,\\ \set(d)\end{array}$} (2)
                   edge [bend right,dashed] node   [below]     {$\begin{array}{c}\varepsilon, \equals(c) ,\\ \set(e)\end{array}$} (3);
  \path[->] (2) edge              node        {$\begin{array}{c}\varepsilon, \equals(d),\\ \set(e)\end{array}$} (3) ;
\end{tikzpicture}
\end{center}\caption{Adding a transition to a non-\good\ TSA.\label{fig:making_standardised}}
\end{figure}
\end{remark}

Note that the $2$-TSA given in Example~\ref{eg:ABCD} is \good\ and accepts only when the pointer is at the root.

\begin{definition}[Proper run] Let  $\mathcal{A}=(Q,\CCC, \Sigma, q_0,  \delta, Q_f)$ be a tree stack automaton.  
We call a valid run     $\mathcal{R}=\tau_1 \cdots \tau_r$  \emph{proper}
    if it does not contain consecutive transitions $\tau_j= (q_1,\varepsilon,p_1,f_1,q_2)$, $\tau_{j+1}=(q_2,\varepsilon,p_2,f_2,q_3)$ with $f_1,f_2\in\{\id,\set(x)\mid x \in C\}$
    (that is, reading $\varepsilon$ with the pointer in the same position for two steps).
\end{definition}

It is clear that if $\mathcal A$ is \good, then every $w\in L(\mathcal A)$ is accepted by some proper run. 
Let $L$ be $k$-MCF language. Theorem~\ref{equ} and Remark~\ref{rmk:standardised} imply that there exists a \good\ $k$-restricted tree stack automaton $\mathcal{A}$ accepting $L$. We fix such a TSA which we denote as
$\mathcal{A}=(Q,\CCC, \Sigma, q_0,  \delta, Q_f)$  in what follows.

\begin{definition}[\Dendron]
Let $\mathcal A$ be a TSA, $\mathcal R$ a proper run of $\mathcal A$ accepting a word $w$, and  $(T_0,\varepsilon)$ the tree stack with pointer at the root after $\mathcal R$ has been performed. Call $T_0$ the \emph{\dendron} corresponding to $\mathcal R$.
\end{definition}

The next technical definition gives a notation to record states and tree-labels associated with a pointer visiting a node from below during a run of a TSA.

\begin{definition}[\Updownvector\ at at node $\nu$]\label{deflm} 
Let $\mathcal A$ be a TSA and $w \in L(\mathcal A)$ be a non-empty word  accepted by the proper run $\mathcal{R}=\tau_1 \cdots \tau_r$, where  for $1\le i\le r$, 
    $\t_i=(q_i, a_i, p_i, f_i,q'_i)$,  $q_i,q_i'\in Q$, $a_i\in \Sigma_{\varepsilon}$,  $p_i$ is a predicate and $f_i$ is a partial  function.
Let
   $(\xi_j,\rho_j)$ denote the tree stack after applying $\tau_1\cdots \tau_j$ for $1\le i\le r$,  $(\xi_0,\rho_0)=(\varepsilon,@)$,  and $T=\xi_r$ denote the \dendron\ corresponding to $\cR$. 
 Suppose $\nu, \nu'$ are addresses in $T$ with  $\nu'\in \N^*$ and $\nu = \nu'n$,  such that $\nu$ is visited from below exactly $s \in [1,k]$ times by $\mathcal R$.  Recall that $\nu$ has exactly one node below it in $T$, which is $\nu'$, so there exists a sequence of indices 
      $1\leq \ell_1\leq m_1 < \ell_2 \leq m_2 < \dots < \ell_s \leq m_s < r$ 
    such that 
 \begin{\arabiclist}
 \item\label{it:uptov} $\rho_{x-1}=\nu', \rho_{x}=\nu$  if and only if $x=\ell_j$ for some $j\in[1,s]$;  
\item\label{it:downfromv} $\rho_{y}=\nu, \rho_{y+1}=\nu'$, if and only if $y=m_j$ for some $j\in[1,s]$; 
\item \begin{alphlist}[(a)]
    \item $f_{\ell_1}=\push_n(c_{\ell_1})$, where $c_{\ell_1}\in \CCC$,
\item $f_{\ell_j}= \up_n$ for $j\in[2,s]$ and
\item $f_{m_j+1} = \down$ for $j\in[1,s]$.
 \end{alphlist}
\end{\arabiclist} 
Then $(\ell,m)_{\mathcal R, \nu} = (\ell_1,m_1,\dots, \ell_s,m_s)$ is called the \emph{\updownvector} of $w$ at $\nu$ with respect to $\mathcal R$.\\[1em]
%We call $\tau_{\ell_1}, \tau_{m_1+1}, \ldots, \tau_{\ell_s}, \tau_{m_s+1}$ the \emph{\tauLMtransitions} of $\cR$ at $\nu$.  MEE THIS NOTATION IS NEVER USED LATER.
\end{definition}

\begin{remark}\label{rem:abovev} We make the following observations.
  \begin{enumerate}
  \item The $\t_{\ell_j}$ transitions of  the \updownvector\ $(\ell,m)_{\mathcal R, \nu}$ of the proper run $\cR$ at $\nu$ are precisely those that
move the pointer \emph{up} from $\nu'$ to $\nu$, via either push or up partial functions.
The $\t_{m_j}$ transitions of  the \updownvector\ are precisely those that, after their application, result in the pointer being at $\nu$ immediately before it
is moved \emph{down} from $\nu$ to $\nu'$, via the down partial function of $\t_{m_{j+1}}$. Since,
by parts \ref{it:uptov} and \ref{it:downfromv} of the definition, there
are no other transitions of $\cR$ with either of these properties, it follows that the pointer remains above the vertex $\nu$ (according to
Definition \ref{defn:below-above}) 
from the start to finish  of  subsequences $\t_{\ell_j+1}\cdots \t_{m_j}$ of $\cR$, for $1\le j\le s$, and remains below $\nu$  at all other times. 
\item It is clear that $m_j < \ell_{j+1}$ for $j\in[1,s-1]$ since a single transition cannot simultaneously move the pointer up and down.
\item If $\ell_1=1$ then $\nu'=\varepsilon$ is the root node.
\item   If $\ell_j=m_j$ for some $j\in[1,s]$,  the transition $\tau_{\ell_j}$ moves the pointer from $\nu'$ to $\nu$ then $\tau_{m_j+1}=\tau_{\ell_j+1}$  immediately moves the pointer back down to $\nu'$.
\end{enumerate}
\end{remark}

For the purposes of the following definitions,  let $(l,m)_{\cR,\nu}=(\ell_1,m_1,\dots, \ell_s,m_s)$ be the \updownvector\  
  of a non-empty word $w \in L$, with respect to a proper run $\mathcal R$ of $\mathcal A$ at a (non-root) vertex $\nu$ of $T$.
\begin{definition}[$\nu$-factorisation]\label{defnuf}
Define the \emph{$\nu$-factorisation} of $w$, with respect to $\mathcal R$, to be \[[w]_{\mathcal R, \nu}=w_0u_1w_1 \cdots u_s w_s,\] where 
\begin{\arabiclist}
\item $w_0$ is the prefix of $w$ that is read by $\tau_1 \cdots \tau_{\ell_1}$, 
\item if $\ell_j=m_j$ then  $u_j=\varepsilon$, else $u_j$  is the factor of $w$ that is read by $\tau_{\ell_j+1} \cdots \tau_{m_j}$ for $j \in[1,s]$, 
\item $w_j$ is the factor of $w$ that is read by $\tau_{m_j+1} \cdots \tau_{\ell_{(j+1)}}$ for $j \in [1,s-1]$,
\item $w_s$ is the factor of $w$ that is read by $\tau_{m_s+1} \cdots \tau_r.$
\end{\arabiclist}
\end{definition}

\begin{remark} From Remark \ref{rem:abovev},  the factors $w_i$ in a $\nu$-factorisation represent letters of the input word which (in the terminology of  Definition~\ref{defn:below-above}) are read by transitions where the pointer starts or finishes \emph{at or below} the vertex $\nu'$  of the tree stack (as it is being constructed), and the factors $u_i$ represent letters read when the pointer starts and finishes \emph{at or above} $\nu$. 
For example, transitions $\tau_1,\dots,\tau_{\ell_1-1}$ move the pointer from the root around nodes which are \emph{below} $\nu$, $\tau_{\ell_1}$ starts with the pointer below $\nu$ and finishes with the pointer at $\nu$, 
if $\ell_1\neq m_1$, the transitions $\tau_{\ell_1+1} \cdots    \tau_{m_1}$ move the pointer around nodes \emph{above} $\nu$
and $\tau_{m_1+1}$ finishes with the pointer below $\nu$. 
 See Fig.~\ref{fig:nu-fact} for a visual representation. 

\begin{figure}[!htb]
  \[\begin{array}{c|c|c|c|c|c|cccccc}
  \tau_1  \cdots  \tau_{\ell_1} & \tau_{\ell_1+1} \cdots    \tau_{m_1} &  \tau_{m_1+1}  \cdots   \tau_{\ell_2} & \cdots 
  &  \tau_{m_{(s-1)}+1}  \cdots 
    \tau_{\ell_s} &  \tau_{\ell_s+1}  \cdots   \tau_{m_s} &  \tau_{m_s+1}  \cdots   \tau_r\\
   w_0 &  u_1 &w_ 1 &\cdots & w_{s-1}   & u_s & w_s
  \end{array}
  \]

  \caption{Schematic of a $\nu$-factorisation of $w$;  $u_j$ factors are possibly empty.  \label{fig:nu-fact}}
  \end{figure}
\end{remark}

The idea for the next definition is to record, at each non-root vertex of the tree stack, a list of 
states and labels that occur each time the pointer is at that vertex, during a proper run accepting a word $w$. This list will be encoded as a
``\emph{\historyvec}'' at the vertex.

  \begin{definition}[\Historyvec]\label{defn:HistVect}
  Let $\nu$ be a non-root vertex,  $(\ell,m)_{\mathcal R, \nu} = (\ell_1,m_1,\dots, \ell_s,m_s)$ be  the \updownvector\ of $w$ at $\nu$ with respect to a proper run $\mathcal R=\tau_1 \cdots \tau_r$  where $\tau_z =(q_z, a_z, p_z , f_z, q'_z) $,   $q_z, q'_z \in Q$  for $z\in[1,r]$.   
  
Define the \emph{\historyvec} associated to $w, \mathcal R$ and $\nu$ to be the $2\times (2s)$ array
\begin{equation*}
\harry_{w, \mathcal{R},\nu}=  
\begin{pmatrix}
c_1^u & c_1^d& \cdots & c_s^u& c_s^d\\
q_1^u & q_1^d & \cdots & q_s^u & q_s^d
\end{pmatrix}
\end{equation*}
where for $j\in[1,s]$:
 \begin{\arabiclist} 
 \item  $c_j^u$ is the  label of the address $\nu$ in the tree stack $(\xi_{\ell_j}, \rho_{\ell_j}=\nu)$  when the pointer has just moved from $\nu'$ up to  $\nu$ via the transition $\tau_{\ell_j}$,
  \item  $c_j^d$ is the  label of the address $\nu$ in the tree stack $(\xi_{m_j}, \rho_{m_j}=\nu)$  just before the pointer moves from $\nu$ down to  $\nu'$ by the transition $\tau_{m_j+1}$, 
\item $q_j^u=q'_{\ell_j}$  is the state of the TSA when the pointer has just moved from $\nu'$ up to  $\nu$,
\item $q_j^d=q'_{m_j}$  is the state of the TSA just before the pointer moves from $\nu$ down to  $\nu'$.
\end{\arabiclist}

\end{definition}

The superscripts $u,d$ stand for ``up'' and ``down'', to indicate when these labels and states occur.

\begin{example}[\updownvector, $\nu$-factorisation, \historyvec]
\label{moreDetailEG}
Consider a TSA $\mathcal A$ which admits a run $\mathcal R = \s_1 \cdots \s_{14}$ where 
  \[\begin{array}{llllll}
  \s_1= (q_0,a,\true,\push_1(c_1),q_1),  & \s_6= (q_5,d,\true,\down,q_6),&   \s_{11}= (q_3,f,\true,\push_1(c_7),q_3),
  \\
  \s_2= (q_1,b,\true,\push_1(c_2),q_2), &  \s_7= (q_6,e,\true,\push_2(c_5),q_7),&     \s_{12}= (q_3,\varepsilon,\true,\down,q_0), 
          \\
      
      \s_{3}= (q_2,\varepsilon,\true,\set(c_3),q_3), & \s_8= (q_7,\varepsilon,\true,\down,q_4),&   \s_{13}= (q_0,g,\true,\down,q_7),
    \\
  \s_4= (q_3,c,\true,\push_2(c_4),q_4), &\s_9= (q_4,\varepsilon,\true,\up_1,q_4), &   \s_{14}= (q_7,h,\true,\down,q_8). 
      \\
  \s_5= (q_4,\varepsilon,\true,\down,q_5), & \s_{10}= (q_4,\varepsilon,\true,\set(c_6),q_3),   
  \end{array}\]
 accepting the word $abcdefgh$ whilst constructing the tree stack   shown in Fig.~\ref{fig:TreeStackEG2DETAILS}.

\begin{figure}[h]   

\begin{tabular}{|c|c|c|}
\hline
\begin{tikzpicture}[scale=.8]
\draw (2,0) node {$\bullet$};
\draw (2.1,0.3) node {$@$};
\draw[decorate,ultra thick,->] (1,0) -- (1.5,0);
\draw (1.85,-0.15) node {\color{cyan}\tiny$\varepsilon$};
\end{tikzpicture}

&
$\sigma_1$:
\begin{tikzpicture}[scale=.8]
\draw[decorate]  (1,1) -- (2,0);
\draw (2,0) node {$\bullet$};
\draw (1,1) node {$\bullet$};
\draw (2.1,0.3) node {$@$};
\draw (1,1.3) node {$c_1$};
\draw[decorate,ultra thick,->] (0,1) -- (0.5,1);
\draw (0.85,0.85) node {\color{cyan}\tiny$1$};
\draw (1.85,-0.15) node {\color{cyan}\tiny$\varepsilon$};
\end{tikzpicture}

& $\sigma_2$\tiny$=\tau_{\ell_1}$\normalsize \begin{tikzpicture}[scale=.8]
\draw[thick,magenta]  (0,2) --  (1,1);
\draw[decorate]   (1,1) -- (2,0);
\draw (0,2) node {$\bullet$};
\draw (1,1) node {$\bullet$};
\draw (2,0) node {$\bullet$};
\draw (2.1,0.3) node {$@$};
\draw (0,2.3) node {$c_2$};
\draw (1,1.3) node {$c_1$};
\draw[decorate,ultra thick,->] (-1,2) -- (-.5,2);
\draw (-.15,1.85) node {\color{cyan}\tiny$11$};
\draw (0.85,0.85) node {\color{cyan}\tiny$1$};
\draw (1.85,-0.15) node {\color{cyan}\tiny$\varepsilon$};
\end{tikzpicture}

\\
\hline $\sigma_3$:\begin{tikzpicture}[scale=.8]
\draw[thick,magenta]  (0,2) --  (1,1);
\draw[decorate]   (1,1) -- (2,0);
\draw (0,2) node {$\bullet$};
\draw (1,1) node {$\bullet$};
\draw (2,0) node {$\bullet$};
\draw (2.1,0.3) node {$@$};
\draw (0,2.3) node {{\color{red}$c_3$}};
\draw (1,1.3) node {$c_1$};
\draw[decorate,ultra thick,->] (-1,2) -- (-.5,2);
\draw (-.15,1.85) node {\color{cyan}\tiny$11$};
\draw (0.85,0.85) node {\color{cyan}\tiny$1$};
\draw (1.85,-0.15) node {\color{cyan}\tiny$\varepsilon$};
\end{tikzpicture} 

& $\sigma_4$:\begin{tikzpicture}[scale=.8]
\draw[thick,magenta]  (0,2) --  (1,1);
\draw[decorate]   (1,1) -- (2,0);
\draw[decorate] (0,2) -- (1,3);
\draw (0,2) node {$\bullet$};
\draw (1,1) node {$\bullet$};
\draw (2,0) node {$\bullet$};
\draw (1,3) node {$\bullet$};
\draw (2.1,0.3) node {$@$};
\draw (0,2.3) node {{$c_3$}};
\draw (1,1.3) node {$c_1$};
\draw (1.1,3.3) node {$c_4$};
\draw[decorate,ultra thick,->] (0,3) -- (0.5,3);
\draw (-.15,1.85) node {\color{cyan}\tiny$11$};
\draw (0.85,0.85) node {\color{cyan}\tiny$1$};
\draw (1.15,2.85) node {\color{cyan}\tiny$112$};
\draw (1.85,-0.15) node {\color{cyan}\tiny$\varepsilon$};
\end{tikzpicture}

& $\sigma_5$ \tiny $=\tau_{m_1}$ \normalsize \begin{tikzpicture}[scale=.8]
\draw[thick,magenta]  (0,2) --  (1,1);
\draw[decorate]   (1,1) -- (2,0);
\draw[decorate] (0,2) -- (1,3);
\draw (0,2) node {$\bullet$};
\draw (1,1) node {$\bullet$};
\draw (2,0) node {$\bullet$};
\draw (1,3) node {$\bullet$};
\draw (2.1,0.3) node {$@$};
\draw (0,2.3) node {{$c_3$}};
\draw (1,1.3) node {$c_1$};
\draw (1.1,3.3) node {$c_4$};
\draw[decorate,ultra thick,->] (-1,2) -- (-.5,2);
\draw (-.15,1.85) node {\color{cyan}\tiny$11$};
\draw (0.85,0.85) node {\color{cyan}\tiny$1$};
\draw (1.15,2.85) node {\color{cyan}\tiny$112$};
\draw (1.85,-0.15) node {\color{cyan}\tiny$\varepsilon$};
\end{tikzpicture}

\\
\hline $\sigma_6$ \tiny $=\tau_{m_1+1}$ \normalsize  \begin{tikzpicture}[scale=.8]
\draw[thick,magenta]  (0,2) --  (1,1);
\draw[decorate]   (1,1) -- (2,0);
\draw[decorate] (0,2) -- (1,3);
\draw (0,2) node {$\bullet$};
\draw (1,1) node {$\bullet$};
\draw (2,0) node {$\bullet$};
\draw (1,3) node {$\bullet$};
\draw (2.1,0.3) node {$@$};
\draw (0,2.3) node {{$c_3$}};
\draw (1,1.3) node {$c_1$};
\draw (1.1,3.3) node {$c_4$};
\draw[decorate,ultra thick,->] (0,1) -- (0.5,1);
\draw (-.15,1.85) node {\color{cyan}\tiny$11$};
\draw (0.85,0.85) node {\color{cyan}\tiny$1$};
\draw (1.15,2.85) node {\color{cyan}\tiny$112$};
\draw (1.85,-0.15) node {\color{cyan}\tiny$\varepsilon$};
\end{tikzpicture}

&$\sigma_7$:\begin{tikzpicture}[scale=.8]
\draw[thick,magenta]  (0,2) --  (1,1);
\draw[decorate]   (1,1) -- (2,0);
\draw[decorate] (1,1) -- (2,2);
\draw[decorate] (0,2) -- (1,3);
\draw (0,2) node {$\bullet$};
\draw (1,1) node {$\bullet$};
\draw (2,2) node {$\bullet$};
\draw (2,0) node {$\bullet$};
\draw (1,3) node {$\bullet$};
\draw (2.1,0.3) node {$@$};
\draw (0,2.3) node {{$c_3$}};
\draw (2.1,2.3) node {$c_5$};
\draw (1,1.3) node {$c_1$};
\draw (1.1,3.3) node {$c_4$};
\draw[decorate,ultra thick,->] (1,2) -- (1.5,2);
\draw (-.15,1.85) node {\color{cyan}\tiny$11$};
\draw (0.85,0.85) node {\color{cyan}\tiny$1$};
\draw (2.15,1.85) node {\color{cyan}\tiny$12$};
\draw (1.15,2.85) node {\color{cyan}\tiny$112$};
\draw (1.85,-0.15) node {\color{cyan}\tiny$\varepsilon$};
\end{tikzpicture}

&
$\sigma_8$:\begin{tikzpicture}[scale=.8]
\draw[thick,magenta]  (0,2) --  (1,1);
\draw[decorate]   (1,1) -- (2,0);
\draw[decorate] (1,1) -- (2,2);
\draw[decorate] (0,2) -- (1,3);
\draw (0,2) node {$\bullet$};
\draw (1,1) node {$\bullet$};
\draw (2,2) node {$\bullet$};
\draw (2,0) node {$\bullet$};
\draw (1,3) node {$\bullet$};
\draw (2.1,0.3) node {$@$};
\draw (0,2.3) node {{$c_3$}};
\draw (2.1,2.3) node {$c_5$};
\draw (1,1.3) node {$c_1$};
\draw (1.1,3.3) node {$c_4$};
\draw[decorate,ultra thick,->] (0,1) -- (0.5,1);
\draw (-.15,1.85) node {\color{cyan}\tiny$11$};
\draw (0.85,0.85) node {\color{cyan}\tiny$1$};
\draw (2.15,1.85) node {\color{cyan}\tiny$12$};
\draw (1.15,2.85) node {\color{cyan}\tiny$112$};
\draw (1.85,-0.15) node {\color{cyan}\tiny$\varepsilon$};
\end{tikzpicture}

\\
\hline

$\sigma_9$\tiny $=\tau_{\ell_2}$ \normalsize \begin{tikzpicture}[scale=.8]
\draw[thick,magenta]  (0,2) --  (1,1);
\draw[decorate]   (1,1) -- (2,0);
\draw[decorate] (1,1) -- (2,2);
\draw[decorate] (0,2) -- (1,3);
\draw (0,2) node {$\bullet$};
\draw (1,1) node {$\bullet$};
\draw (2,2) node {$\bullet$};
\draw (2,0) node {$\bullet$};
\draw (1,3) node {$\bullet$};
\draw (2.1,0.3) node {$@$};
\draw (0,2.3) node {{$c_3$}};
\draw (2.1,2.3) node {$c_5$};
\draw (1,1.3) node {$c_1$};
\draw (1.1,3.3) node {$c_4$};
\draw[decorate,ultra thick,->] (-1,2) -- (-.5,2);
\draw (-.15,1.85) node {\color{cyan}\tiny$11$};
\draw (0.85,0.85) node {\color{cyan}\tiny$1$};
\draw (2.15,1.85) node {\color{cyan}\tiny$12$};
\draw (1.15,2.85) node {\color{cyan}\tiny$112$};
\draw (1.85,-0.15) node {\color{cyan}\tiny$\varepsilon$};
\end{tikzpicture}

&$\sigma_{10}$:\begin{tikzpicture}[scale=.8]
\draw[thick,magenta]  (0,2) --  (1,1);
\draw[decorate]   (1,1) -- (2,0);
\draw[decorate] (1,1) -- (2,2);
\draw[decorate] (0,2) -- (1,3);
\draw (0,2) node {$\bullet$};
\draw (1,1) node {$\bullet$};
\draw (2,2) node {$\bullet$};
\draw (2,0) node {$\bullet$};
\draw (1,3) node {$\bullet$};
\draw (2.1,0.3) node {$@$};
\draw (0,2.3) node {{{\color{red}$c_6$}}};
\draw (2.1,2.3) node {$c_5$};
\draw (1,1.3) node {$c_1$};
\draw (1.1,3.3) node {$c_4$};
\draw[decorate,ultra thick,->] (-1,2) -- (-.5,2);
\draw (-.15,1.85) node {\color{cyan}\tiny$11$};
\draw (0.85,0.85) node {\color{cyan}\tiny$1$};
\draw (2.15,1.85) node {\color{cyan}\tiny$12$};
\draw (1.15,2.85) node {\color{cyan}\tiny$112$};
\draw (1.85,-0.15) node {\color{cyan}\tiny$\varepsilon$};
\end{tikzpicture} 

& $\sigma_{11}$:\begin{tikzpicture}[scale=.8]
\draw[thick,magenta]  (0,2) --  (1,1);
\draw[decorate]   (1,1) -- (2,0);
\draw[decorate] (1,1) -- (2,2);
\draw[decorate] (-1,3)-- (0,2) -- (1,3);
\draw (0,2) node {$\bullet$};
\draw (1,1) node {$\bullet$};
\draw (2,2) node {$\bullet$};
\draw (2,0) node {$\bullet$};
\draw (-1,3) node {$\bullet$};
\draw (1,3) node {$\bullet$};
\draw (2.1,0.3) node {$@$};
\draw (0,2.3) node {{{$c_6$}}};
\draw (2.1,2.3) node {$c_5$};
\draw (1,1.3) node {$c_1$};
\draw (-1.1,3.3) node {$c_7$};
\draw (1.1,3.3) node {$c_4$};
\draw[decorate,ultra thick,->] (-2,3) -- (-1.5,3);
\draw (-.15,1.85) node {\color{cyan}\tiny$11$};
\draw (0.85,0.85) node {\color{cyan}\tiny$1$};
\draw (2.15,1.85) node {\color{cyan}\tiny$12$};
\draw (-1.15,2.85) node {\color{cyan}\tiny$111$};
\draw (1.15,2.85) node {\color{cyan}\tiny$112$};
\draw (1.85,-0.15) node {\color{cyan}\tiny$\varepsilon$};
\end{tikzpicture}
\\
\hline
 $\sigma_{12}$\tiny $=\tau_{m_2}$ \normalsize \begin{tikzpicture}[scale=.8]
\draw[thick,magenta]  (0,2) --  (1,1);
\draw[decorate]   (1,1) -- (2,0);
\draw[decorate] (1,1) -- (2,2);
\draw[decorate] (-1,3)-- (0,2) -- (1,3);
\draw (0,2) node {$\bullet$};
\draw (1,1) node {$\bullet$};
\draw (2,2) node {$\bullet$};
\draw (2,0) node {$\bullet$};
\draw (-1,3) node {$\bullet$};
\draw (1,3) node {$\bullet$};
\draw (2.1,0.3) node {$@$};
\draw (0,2.3) node {{{$c_6$}}};
\draw (2.1,2.3) node {$c_5$};
\draw (1,1.3) node {$c_1$};
\draw (-1.1,3.3) node {$c_7$};
\draw (1.1,3.3) node {$c_4$};
\draw[decorate,ultra thick,->] (-1,2) -- (-.5,2);
\draw (-.15,1.85) node {\color{cyan}\tiny$11$};
\draw (0.85,0.85) node {\color{cyan}\tiny$1$};
\draw (2.15,1.85) node {\color{cyan}\tiny$12$};
\draw (-1.15,2.85) node {\color{cyan}\tiny$111$};
\draw (1.15,2.85) node {\color{cyan}\tiny$112$};
\draw (1.85,-0.15) node {\color{cyan}\tiny$\varepsilon$};
\end{tikzpicture}

&

$\sigma_{13}$\tiny $=\tau_{m_2+1}$ \normalsize \begin{tikzpicture}[scale=.8]
\draw[thick,magenta]  (0,2) --  (1,1);
\draw[decorate]   (1,1) -- (2,0);
\draw[decorate] (1,1) -- (2,2);
\draw[decorate] (-1,3)-- (0,2) -- (1,3);
\draw (0,2) node {$\bullet$};
\draw (1,1) node {$\bullet$};
\draw (2,2) node {$\bullet$};
\draw (2,0) node {$\bullet$};
\draw (-1,3) node {$\bullet$};
\draw (1,3) node {$\bullet$};
\draw (2.1,0.3) node {$@$};
\draw (0,2.3) node {{{$c_6$}}};
\draw (2.1,2.3) node {$c_5$};
\draw (1,1.3) node {$c_1$};
\draw (-1.1,3.3) node {$c_7$};
\draw (1.1,3.3) node {$c_4$};
\draw[decorate,ultra thick,->] (0,1) -- (0.5,1);
\draw (-.15,1.85) node {\color{cyan}\tiny$11$};
\draw (0.85,0.85) node {\color{cyan}\tiny$1$};
\draw (2.15,1.85) node {\color{cyan}\tiny$12$};
\draw (-1.15,2.85) node {\color{cyan}\tiny$111$};
\draw (1.15,2.85) node {\color{cyan}\tiny$112$};
\draw (1.85,-0.15) node {\color{cyan}\tiny$\varepsilon$};
\end{tikzpicture} 

& 
$\sigma_{14}$:
\begin{tikzpicture}[scale=.8]
\draw[thick,magenta]  (0,2) --  (1,1);
\draw[decorate]   (1,1) -- (2,0);
\draw[decorate] (1,1) -- (2,2);
\draw[decorate] (-1,3)-- (0,2) -- (1,3);
\draw (0,2) node {$\bullet$};
\draw (1,1) node {$\bullet$};
\draw (2,2) node {$\bullet$};
\draw (2,0) node {$\bullet$};
\draw (-1,3) node {$\bullet$};
\draw (1,3) node {$\bullet$};
\draw (2.1,0.3) node {$@$};
\draw (0,2.3) node {$c_6$};
\draw (2.1,2.3) node {$c_5$};
\draw (1,1.3) node {$c_1$};
\draw (-1.1,3.3) node {$c_7$};
\draw (1.1,3.3) node {$c_4$};
\draw[decorate,ultra thick,->] (1,0) -- (1.5,0);
\draw (-.15,1.85) node {\color{cyan}\tiny$11$};
\draw (0.85,0.85) node {\color{cyan}\tiny$1$};
\draw (2.15,1.85) node {\color{cyan}\tiny$12$};
\draw (-1.15,2.85) node {\color{cyan}\tiny$111$};
\draw (1.15,2.85) node {\color{cyan}\tiny$112$};
\draw (1.85,-0.15) node {\color{cyan}\tiny$\varepsilon$};
\end{tikzpicture} \\
\hline
\end{tabular}

\caption{Tree stack shown transition-by-transition 
in Example~\ref{moreDetailEG}, giving \updownvector\ $(2,5,9,12)$. The edge from $\nu'=1$ to $\nu=11$ is highlighted, as are labels that are replaced using a $\set(c_i)$ move.
\label{fig:TreeStackEG2DETAILS}}
\end{figure}

Choosing $\nu=11$ we obtain \updownvector\ $(2,5,9,12)$
and 
$\nu$-factorisation \[[w]_{\s_1 \cdots \s_{14}, 11}= ab \cdot c \cdot de \cdot f \cdot gh\] (see Fig.~\ref{fig:TreeStackEG2DETAILS}), via:
  \[\begin{array}{c|c|c|c|ccccccc}
  \s_{1}  \s_2 &  \s_3 \s_4 \s_5 & \s_6 \s_7  \s_8 \s_9 & \s_{10} \s_{11} \s_{12} & \s_{13} \s_{14}\\
  ab & c &de &f  &gh \\
  \end{array}.
\]

The \historyvec\
  $\harry_{abcdefgh,\s_1 \cdots \s_{14},11}$ is 
 \[\left(\begin{array}{lllllll}
 c_2 & c_3 & c_3 & c_6 \\
 q_{2} & q_{5} &q_{4} & q_{0}
 \end{array}\right).\] 
  \end{example}

\begin{definition}[Up-set]\label{defU}   Let $L$ be a $k$-MCF language accepted by $\mathcal{A}$. Let  $\mathbf{c}=(c_1,c'_1,\dots,c_s,c'_s)\in \CCC^{2s}$ and
  $\mathbf{q}=(q_1,q'_1,\dots, q_s,q'_s)\in Q^{2s}$, for some $s\in[1, k]$. 
    Define the \emph{\UpSet} associated to
  $\candq$, denoted 
  $U_{\candq}$, to be the set consisting of all $(u_1,u_2,\dots,u_s)\in (\Sigma^*)^s$ such that there exists $w\in L$, 
  a proper run $\mathcal R$ accepting $w$ with the \dendron\ $T$ 
   and an address $\nu$ in $T$ such that $w$ has a  factorisation 
\begin{equation*}
    [w]_{\mathcal{R},\nu}= w_0  u_1 w_1  \cdots u_s w_s,
\end{equation*}
with corresponding \historyvec\
\begin{equation*}
\harry_{w, \mathcal{R},\nu}= \begin{pmatrix} \mathbf{c}\\ \mathbf{q}\end{pmatrix} =
\begin{pmatrix}
c_1 & c'_1& \cdots & c_s& c'_s\\
q_1 & q'_1 & \cdots & q_s & q'_s
\end{pmatrix}
.
\end{equation*}
\end{definition}

Note that since $s\leq k$ and $|C|,|Q|$ are finite, there are finitely many $\candq$ and hence finitely many sets $U_{\candq}$. The size of each \UpSet\ (the number of tuples it contains) can be finite or infinite.
For $(\mathbf{c},\mathbf{q})\in \cup_{i=1}^k \CCC^{2i}\times Q^{2i}$ such that   $U_{\candq}$ is finite  let
  $\mathcal M_{\mathbf{c},\mathbf{q}} = \max\{(|u_1| + \ldots + |u_s|) \,\colon\, (u_1,u_2,\dots,u_s) \in U_{\candq}\}$  and
\begin{align}\label{eq:M}
  \mathcal M  =\begin{cases}
   \max\{\cM_{\mathbf{c},\mathbf{q}}\,:\,  (\mathbf{c},\mathbf{q})\in \cup_{i=1}^k \CCC^{2i}\times Q^{2i}\textrm{ and }
  U_{\candq}\textrm{ is finite}\}\\
  0  \quad \text{if all \UpSet s are infinite}.
  \end{cases}
\end{align}

The next lemma is a simple application of \UpSet s and \historyvec s.
  \begin{lemma} [Single swap] \label{1s}
Let $L$ be a $k$-MCF language accepted by $\mathcal{A}$. For $w\in L$, let $\mathcal{R}$ be  a proper run accepting $w$ with the corresponding \dendron\ ${T}$. Choose a non-root vertex $\nu$ 
in $T$, and let $[w]_{\mathcal{R},\nu}$ and $\harry_{w,\mathcal{R},\nu}$ be the corresponding factorisation and \historyvec s.
If there exists some $w'\in L$ accepted by a proper run $\mathcal{R}'$, producing \dendron\ ${T}'$ and vertex $\nu'\in {T}'$ so that
\begin{equation*}
   [w']_{\mathcal{R}',\nu'}= w_0' u_1' w_1'  \cdots u_s' w_s',
\end{equation*}
and 
\begin{equation*}
    \harry_{w',\mathcal{R}',\nu'}=\harry_{w,\mathcal{R},\nu},
\end{equation*}
then 
\begin{equation*}
    w''=w_0u_1'w_1u_2'w_2\dots u_s'w_s\in L.
\end{equation*}
\end{lemma}

\begin{proof}
  Let the up-down vector and \historyvec\ of $w$ at $\nu$  with respect to $\cR$ be \[(l,m)_{\cR,\nu}=(l_1,m_1,\ldots, l_s,m_s)\] and 
  \[ 
\harry_{w, \mathcal{R},\nu}=
\begin{pmatrix}
c^u_1 & c^d_1& \cdots & c^u_s& c^d_s\\
q^u_1 & q^d_1 & \cdots & q^u_s & q^d_s
\end{pmatrix}= \harry_{w',\mathcal{R}',\nu'}\] respectively, and let the up-down vector of $w'$ at $\nu'$ with respect to $\cR'$ be
\[(l,m)_{\cR',\nu'}=(l'_1,m'_1,\ldots, l'_s,m'_s).\] 
We can describe an accepting run for $w''$ as follows. Read $w_0$ by following 
$\mathcal R$ until the pointer is at $\nu$ for the first time, when $\mathcal A$ is in state $q_1^u$ and the label of $\nu$ is $c_1^u$.
Then follow the factor $\mathcal R'$ beginning with $\tau_{l_1'+1}$, which reads $u_1'$ with $\mathcal A$ starting in state $q_1^u$ and the pointer at $\nu$ labeled $c_1^u$, building the first part of a copy of a subtree of $T'$ above $\nu$,
until the transition $\tau_{m_1'}$ of $\cR'$, when the pointer is at $\nu$ and about to move down. At this point
 the node $\nu$ is labeled $c_1^d$ and $\mathcal A$  is in state $q_1^d$. 
 Then follow the factor of $\mathcal R$, starting with $\tau_{m_1+1}$, which reads $w_1$ starting in state $q_1^d$ with the pointer at $\nu$ labeled $c_1^d$, moving around and building more of the tree below $\nu$, until the pointer returns to $\nu$ for the second time, with $\mathcal A$ in  state $q_2^u$ and the label of $\nu$ equal to $c_2^u$.

Continue switching between the runs $\mathcal R$ and $\mathcal R'$, where each switch is possible since the state of $\mathcal A$ and label of $\nu$ match;  the tree stack above $\nu$ corresponds to the tree stack for $w'$ above $\nu'$, and the tree stack below $\nu$ corresponds to the tree stack for $w$ below $\nu$. 

Note that the tree stack built by this process has addresses 
\begin{center}
\begin{tabular}{ll}
$\mu$ & if $\mu$ is a vertex of $T$ and $\nu$ is not a prefix of $\mu$\\
$\nu\mu'$ & if $\nu'\mu'$ is a vertex of $T'$. \\
 \end{tabular}
\end{center}
\end{proof}

\section{Substitution Lemma}

In this section, we state and prove the Substitution Lemma for infinite multiple context-free languages.
We assume further all tree stack automata have at least one push command and therefore have positive degree. 
We shall use the following lemma. In the statement of the lemma, a word $u \in \Sigma^+$ is said to be
 \emph{read by a run of a TSA while the pointer is at a vertex $\nu$ of a tree stack}
  if each letter of $u$ is read by a transition which does not move the pointer and the pointer is at $\nu$ before and after
  applying each transition.

\begin{lemma}\label{lem:atv_2} Let  $L$ be an infinite  $k$-MCF language accepted by $\cA$ and let $D$ be the degree of $\cA$. Let $\cR$
  be a proper run accepting a word $w$, $T$ be the corresponding \dendron\, and $\nu$ be a vertex of $T$. 
   Let $\mu\in \N_+$ and assume that the following condition holds:
  \begin{align}\label{condition}
    \text{ if } w' \text{is a factor of } & w \text{ such that the pointer remains at } \nu \nonumber \\
     \text{ while } \cR \text{ is reading } & w', \text{ then } |w'|\le  \mu|C||Q|.
  \end{align}
    Then the number $n_{w,\nu}$ of letters of $w$ read by $\cR$ while the pointer is at $\nu$ is bounded by
  \begin{itemize}
	\item $n_{w,\nu}\leq \mu k|C||Q|(D+1)$ if  $\nu$ is a non-root, non-leaf vertex of $T$;
	\item $n_{w,\nu}\leq \mu kD|Q|$  if $\nu=\varepsilon$ is the root vertex of $T$;
	\item $n_{w,\nu}\leq \mu k|C||Q|$ if $\nu$ is a leaf of $T$.
  \end{itemize}
\end{lemma}
\begin{proof}
  From the  definitions, $\nu$ has at most $D$ children and (at most) one parent. As $\cA$ is $k$-restricted
the number of transitions of $\cR$ which move the pointer to $\nu$ from below (from its parent) is at most $k$; and the
number of transitions that move the pointer to $\nu$ from one of its children is at most $kD$. Let $r_1, \ldots , r_t$ be the maximal subsequences
of the sequence of transitions $\cR$ such that the pointer is kept at $\nu$ while executing subsequence $r_i$, for $i=1,\ldots, t$.
Then $t \leq k(D+1)$.
By assumption and the fact that $\cR$ is proper, none of these subsequences $r_i$ has length greater than $\mu|C||Q|$.
 Therefore the number of letters of $w$ read by transitions of $\cR$ with the pointer at $\nu$ is at most $\mu k|C||Q|(D+1)$. Note also that the root has no parent and always has label $@$ so if $\nu=\varepsilon$ then this bound is $\mu k D|Q|$.  Furthermore, if $\nu$ is a leaf of $T$, so has no children,  then the bound is $\mu k|C||Q|$.
\end{proof}

\begin{definition}[Pumpable words]
   Let  $L$ be a language over an alphabet $\Sigma$ and let $\alpha,\beta \in \N_+$ with $\alpha \le \beta$. 
A word $w\in L$ is called {\em $(\alpha,\beta)$-pumpable}
if there exists $x,y,z\in \Sigma^*$ with $|y|\in[\alpha,\beta]$
such that $w=xyz$ and for all $n\in\N_0$, $xy^nz\in L$.
\end{definition}

The next lemma says that if, in a given run, the pointer remains at a single vertex for sufficiently many transitions, then this run of the  TSA behaves
  like a run of a finite state automaton and an analogue of the pumping lemma for regular languages applies to the word accepted.  Recall that by Remark~\ref{rmk:node-label-non-empty} we assume $|C|>0$ for any TSA $\cA$.
\begin{lemma}\label{lem:pump_2} Let  $L$ be an infinite  $k$-MCF language accepted by $\cA$. Let $\cR$
  be a proper run accepting a word $w$, 
  let $T$ be the corresponding \dendron\, and let $\nu$ be a vertex of $T$. 
  Assume that, for some  $m\ge 1$, there is a sequence $\hat{r}$ of more than $m|C||Q|$ consecutive transitions of $\cR$,
  during which the pointer remains at $\nu$. Then $w$ is $(m,  m|C||Q|)$-pumpable.
\end{lemma}

\begin{proof}
  We may assume, by considering a prefix if necessary, that $\hat{r}$ consists of precisely
  $p=m|C||Q|+1$ transitions, 
  and write $\cR=r'\hat{r}r''$, where for some $a\ge 1$, $r'=\t_1\cdots \t_{a-1}$, $\hat{r}=\t_a\cdots \t_{a+p-1}$, and $r''=\t_{a+p}\cdots \t_r$, with $\t_1, \ldots, \t_r \in \delta$.
   The generalised pigeon hole principle implies that there are (at least) $m+1$ %integers
  indices  $j_1, \ldots, j_{m+1}$, such that
  $a\le j_1< \cdots< j_{m+1}\le a+ p-1$  and an element $(c,q)\in C\times Q$ 
  such that, for all $i\in[1,m+1]$, after applying $\t_{j_i}$ the label of the vertex $\nu$ is $c$ and the state of $\cA$ is $q$.
  
We may now rewrite $\hat{r}$ as $\hat{r}= r_0 r_1 r_2$, where $r_0 =\t_a\cdots\t_{j_1}$, $r_1 =\t_{j_1+1}\cdots\t_{j_{m+1}}$ and
$r_2 = \t_{j_{m+1}+1}\cdots\t_{a+p-1}$. Denote by $w_0$ the prefix of $w$ read by $r'r_0$,
by $y$ the subword of $w$ read by $r_1$ and by $w_1$ the suffix of $w$ read by $r_2 r''$. 
 By construction, after either of the transitions $\t_{j_1}$ and  $\t_{j_{m+1}}$ have been applied, the pointer is at $\nu$, the label of $\nu$ is $c$ and the state of $\cA$ is $q$.  Hence $\cA$ has a valid run $r'r_0r_2r''$ that accepts the word $w_0w_1$. Similarly, $\cA$ has a valid run $r'r_0r_1^n r_2r''$ that accepts the word $w_0y^nw_1$, for all $n\ge 2$.  That is, $w_0y^nw_1\in L$, for all $n\in \N_0$.
  Moreover, as $\cR$ is proper, $m\le |y|\le p-1$.
\end{proof}

The next definition allows us to refer to the
process of switching between different factorisations of a word $w$
without explicit mention of any particular $\nu$-factorisation.

\begin{definition}[$U$-switchable factorisation]\label{def:Uswitchable}
  Let $L$ be a language over an alphabet $\Sigma$, let $s\in \N^+$ and let $U$ be a set of $s$-tuples of elements of $\Sigma^*$: that is $U\subseteq (\Sigma^*)^s$.
  \begin{\arabiclist}
  \item A $U$\emph{-factorisation} of a word $w\in \Sigma^*$ is a factorisation
      \[w=w_0u_1w_1\cdots  u_s w_s\] 
      such that $(u_1,\ldots ,u_s)\in U$. 
  \item A word $w\in L$ is said to be $U$\emph{-switchable} if $w$ has a $U$-factorisation
    $w=w_0u_1w_1\cdots  u_s w_s$ such that, for all $(v_1,\ldots, v_s)\in U$ the word
    \[w_0v_1w_1\cdots v_s w_s\]
    belongs to $L$. In this case we say the given $U$-factorisation of $w$ is a $U$\emph{-switchable} factorisation. 
\end{\arabiclist}
\end{definition}

Given a word $w$ over an alphabet $\Sigma$ we may distinguish, by marking, some of the letters of $w$. Formally, considering $w$ as a sequence of letters,
we take a subsequence $w'$ and consider its elements \emph{marked}. The number of marked letters of $w$ is then $|w'|$. For notational
convenience we use the $\#$ symbol (which we assume does not belong to $\Sigma$) to represent any marked letter of $w$ and then
write the number of marked letters $|w'|$ of $w$ as  $|w|_\#$. For example, if $w=(ab^3)^5$ and we mark the underlined letters in this expression \[\underline{a}bbba\underline{b}bbabb\underline{b} \, \underline{a}bbb\underline{a}bbb\] for $w$ 
 then $|w|_\#=5$.

\begin{theorem}[Substitution Lemma]  
\label{thm:Substitution} 
  Let  $L$ be an infinite  $k$-MCF language over an alphabet $\Sigma$.
  Given $\mu\in \N_+$ 
   there exist
integers $M, m_1,\dots, m_M\in\N_+$ 
  where $m_i\in[1, k]$ and
 sets $U_1,\dots, U_M$ where  $U_i\subseteq (\Sigma^*)^{m_i}$  for $i\in[1,M]$ 
 such that the following  hold.
\begin{\arabiclist}
\item\label{it:4.1.1} For each $i\in[1,M]$ and each  $(u_1,\dots, u_{m_i})\in U_i$, there exists a word $w\in L$ with $U_i$-switchable
  factorisation \[w=w_0u_1w_1\dots w_{m_i-1}u_{m_i}w_{m_i}.\]
\item\label{it:4.1.2} Given  $\lambda \in \N_+$ 
 there exists $N_\lambda\in\N_+$ such that $N_\lambda\geq \mu$ and if $w\in L$ has  at least $N_{\lambda}$ marked 
 positions,
 then either
 \begin{alphlist}[(a)]
	\item $w$ is $(\mu,N_\lambda)$-pumpable or 
	\item there exists $i\in [1,M]$, such that $U_i$ is infinite,
    and $w$ has a $U_i$-switchable factorisation 
    \[w=w_0u_1w_1\dots w_{m_{i-1}}u_{m_i}w_{m_i}\]
  satisfying \[\Sigma_{j=0}^{m_i}|w_j|_\#\geq \lambda 
  \textrm{ and }\Sigma_{j=1}^{m_i}|u_j|_\#\geq \mu.\]
 \end{alphlist}
\item\label{it:4.1.3} 
  There exists $N_\mu\in \N_+$ such that if $w\in L$ and  $|w|_\#\ge N_\mu$ then
  \begin{alphlist}[(a)]
    \item $w$ is $(\mu,N_\mu)$-pumpable or
    \item there exists $i\in [1,M]$, such that $U_i$ is infinite,   and $w$ has a $U_i$-switchable factorisation 
      \[w=w_0u_1w_1\dots w_{m_{i-1}}u_{m_i}w_{m_i},\]
      satisfying  
      \[\mu\le \Sigma_{j=1}^{m_i}|u_j|_\#\leq N_\mu.\]
  \end{alphlist}
\end{\arabiclist}
\end{theorem}

\begin{proof}
  Since $L$ is $k$-MCF, there exists a \good\ $k$-restricted tree stack automaton $\mathcal{A}=(Q,\CCC, \Sigma, q_0,  \delta, Q_f)$ which accepts $L$. 
  Recalling Definitions~\ref{defn:HistVect} and~\ref{defU}, let $\mathcal C$ be the set of elements
  $(\mathbf c,\mathbf q)\in  \cup_{i=1}^k \CCC^{2i}\times Q^{2i}$  such that there exists $w\in L$, 
  a proper run $\mathcal R$ accepting $w$ with \dendron\ $T$ 
  and an address $\nu$ in $T$ such that the history array $\harry_{w,\cR,\nu}$ is equal to  the $2\times 2i$ array
  $\left(\begin{smallmatrix} \mathbf{c}\\ \mathbf{q} \end{smallmatrix}\right)$. 
  Define \[\mathcal U=\{U_{\candq}\,:\, (\mathbf c,\mathbf q)\in \mathcal C \},\quad
 \mathcal{U}_0=\{U_{\candq}\in\mathcal U\,:\,  U_{\candq}\  \text{is infinite}\},\]
 and  \[\mathcal{U}_1=\{U_{\candq}\in\mathcal U\,:\,  U_{\candq}\  \text{is finite}\}.\]

The set $\cU=\mathcal{U}_0\cup \mathcal{U}_1$ is then finite and we choose
to write it as  $\{ U_1,\dots, U_M\}$. For $U_i\in \mathcal{U}$, let  $m_i$ be the positive integer such that
there is $(\mathbf c,\mathbf q)\in C^{2m_i}\times Q^{2m_i}$ with 
$U_i=U_{\candq}$. 

\smallskip
\noindent (\ref{it:4.1.1}): 
For each $i\in[1,M]$ and each  $(u_1,\dots, u_{m_i})\in U_i$, there exists a word $w\in L$ with $U_i$-factorisation by Definition~\ref{defU} and this factorisation is $U_i$-switchable by Lemma~\ref{1s}.

\smallskip
\noindent(\ref{it:4.1.2}): 
Suppose $\lambda\geq \mu$. 
Let  \[N_{\lambda} = (D+1)(\mu k|\CCC||Q|+\lambda)+D\cM,\] 
where $D = \deg(\mathcal A)$ and $\cM$ is defined by Eq.~\eqref{eq:M}. 
Let $w\in L$, assume at least $N_{\lambda}$ letters of $w$ are marked, 
 let $\mathcal{R}$ be a proper accepting run
 for $w$, and let
 $T$ be the corresponding \dendron. 

 If there exists a vertex $\nu$ of $T$ such  
  that there is a sequence $\hat{r}$ of more than $\mu|C||Q|$ consecutive transitions of $\cR$,
  during which the pointer remains at $\nu$,  then from Lemma~\ref{lem:pump_2}  there exist $w_0,u_1, w_1\in \Sigma^*$ such
  that $w=w_0u_1 w_1$, $\mu\le |u_1|\le \mu|C||Q|$ and, for all $n\in \N_0$, the word
  $w_0u_1^n w_1\in L$. As $\mu |C||Q|\le N_\lambda$,   $w$ is $(\mu, N_\lambda)$-pumpable and (\ref{it:4.1.2}.) holds, in this case.

  We now assume that, for all vertices $\nu$ of $T$, the pointer may remain at $\nu$ for at most  $\mu|C||Q|$ consecutive transitions (so the
 condition~\eqref{condition} of Lemma~\ref{lem:atv_2} holds).
If $\nu$ is a non-root vertex of $T$ then there is a \historyvec\ $\mathbf{h}=\harry_{w,\cR,\nu}$  associated
to $w$, $\cR$ and $\nu$ (see Definition~\ref{defU}) and then a unique integer
$i_{(w,\cR,\nu)}\in [1,M]$ such that the \UpSet\ $U_{\mathbf{h}}$ associated to
$\mathbf{h}$ is equal to $U_{i_{(w,\cR,\nu)}}$. 

We next show the existence of a (non-root) vertex $\nu$ 
 such that, setting $i=i_{(w,\cR,\nu)}$,  the \UpSet\ $U_{i}$ is infinite and the  
 $\nu$-factorisation of $w$ at $\nu$ (see Definition~\ref{defnuf}) has the form $[w]_{\cR,\nu}=w_0u_1\cdots u_{m_i}w_{m_i}$ and
satisfies $\Sigma_{j=0}^{m_i}|w_j|_\#\geq \lambda$
and $\Sigma_{j=1}^{m_i}|u_j|_\#\geq \mu$.

For a vertex $\nu$ of $T$ denote by $T_\nu$ the subtree of $T$ with root $\nu$ consisting of vertices connected
to $\varepsilon$ by a simple path containing $\nu$. That is, $T_\nu$ consists of $\nu$ and all its descendants. 
We call a vertex $\nu$ of $T$ a $\lambda$\emph{-singular vertex} if the automaton $\mathcal{A}$,
while executing $\cR$, reads fewer than $\lambda$ marked letters when the pointer is not at a vertex of $T_\nu$.

As $\lambda \ge 1$, by definition  the root  of $T$ is $\lambda$-singular. Also, if $l$ is a leaf of $T$  then,
by Lemma~\ref{lem:atv_2}, at most $\mu k|C||Q|$  letters of $w$
 can be read by $\cR$ while  the pointer is at $l$, and so if $l$ is $\lambda$-singular then 
 $|w|_\#<  \lambda+\mu k|C||Q|$, a contradiction. Therefore 
no leaf of $T$ is $\lambda$-singular. Moreover, if a vertex is $\lambda$-singular, its parent must be $\lambda$-singular; whence the $\lambda$-singular vertices form a
(non-empty) sub-tree of $T$.   
  From the above, a leaf $l_1$ of the $\lambda$-singular sub-tree cannot be a leaf of $T$, so
every such  $l_1$ is $\lambda$-singular and has $n \ge 1$ children, none of which is $\lambda$-singular.

Now choose a $\lambda$-singular vertex $\nu'$ with a positive number of children, none of which is $\lambda$-singular.
Since $\nu'$ is $\lambda$-singular, the automaton $\mathcal{A}$ reads more than $N_\lambda-\lambda$ marked letters of $w$
when the pointer is at the vertices of $T_{\nu'}$.
From  Lemma \ref{lem:atv_2}, at most $\mu k|C||Q|(D+1)$
 letters of $w$ can be read while the pointer is at $\nu'$ so there must be 
 some child $\nu$ of $\nu'$ such that more than
  $(N_{\lambda}-\lambda-\mu k|C||Q|(D+1))/D = \lambda+\cM$
  marked  letters of $w$ are read while  the pointer is at the vertices of $T_\nu$.  

 Form the factorisation $[w]_{\cR,\nu}=w_0u_1\cdots u_{m_i}w_{m_i}$ of $w$
 at $\nu$ and let $U_{i}$ be the associated \UpSet.
 The $u_j$'s in $[w]_{\cR,\nu}$ are precisely the factors of $w$
 read by $\cA$ while the pointer is at vertices of $T_\nu$. Hence,  
 $\Sigma_{j=1}^{m_i}|u_j|_\#> \lambda+\cM >\mu$, which by Eq.~\eqref{eq:M}
 is possible only for infinite $U_{i}\in \cU_0$.
 
 On the other hand, by assumption $\nu$ is not $\lambda$-singular,  so also $\Sigma_{j=0}^{m_i}|w_j|_\#\geq \lambda$.
 Therefore, in this case, (\ref{it:4.1.2}.) holds.

 Now suppose $\lambda< \mu$. We proved above that (\ref{it:4.1.2}.) holds in the case $\lambda=\mu$
   and using the fact that by definition, for all positive $\lambda'<\mu$, being $(\mu, N_\mu)$-pumpable implies being $(\mu, N_{\lambda'})$-pumpable, it
   can be readily verified that part (\ref{it:4.1.2}.) also holds with $\lambda<\mu$.

\smallskip
\noindent(\ref{it:4.1.3}): 
Let $N_\mu$ be given by the formula for $N_{\lambda}$ above by taking $\lambda=\mu$. Assume $w\in L$,
assume at least $N_{\lambda}$ letters of $w$ are marked  and let $\cR$ and $T$ be defined
as in the proof of (\ref{it:4.1.2}.).
In the case when there is  a vertex $\nu$ of $T$ at which the pointer  remains for more than $\mu|C||Q|$ consecutive transitions,
as in the proof of part (\ref{it:4.1.2}.) above, $w$ is $(\mu, N_\mu)$-pumpable, whence (\ref{it:4.1.3}.) holds in this case.

Now assume there is no vertex $\nu$ such that the pointer remains at $\nu$ for more than $\mu|C||Q|$ consecutive  transitions.
For any vertex $\nu \in {T}$ write $i_\nu$ for $i_{(w,\cR,\nu)}$, $m_\nu$ for $m_{i_\nu}$  and
    the factorisation $[w]_{\cR,\nu}$ of $w$ at $\nu$ as
    \[w=w_0^{(\nu)} u_1^{(\nu)}\cdots u_{m_\nu}^{(\nu)} w_{m_\nu}^{(\nu)}.\]
    If a vertex $\nu$ of $T$ has the property that 
    \[\Sigma_{j=1}^{m_\nu}|u_j^{(\nu)}|_\#> N_\mu= (D+1)(\mu k|\CCC||Q|+\mu)+D\cM,\]
    then, arguing as  in the proof of (\ref{it:4.1.2}.), we see that 
    $\nu$ must have a child $\nu'$  such that more than  $\mathcal{M}+\mu$ marked letters are read by $\cR$ while the pointer is
    at vertices of $T_{\nu'}$. For such $\nu'$ the factorisation $[w]_{\cR,\nu'}$ must satisfy
    \[ \mathcal{M} +\mu <\Sigma_{j=1}^{m_{\nu'}}|u_j^{(\nu')}|_\#.\]
    
    Note that, since $|w|_\#>N_\mu$ and from Lemma~\ref{lem:atv_2} 
    at most $\mu k|Q|D$ letters of $w$ can be read while the pointer is at the root $\varepsilon$,
    there is always some child $\eta$ of the root such that
    the factorisation $[w]_{\cR,\eta}$ satisfies
    \[ \mathcal{M} +\mu <\Sigma_{j=1}^{m_{\eta}}|u_j^{(\eta)}|_\#.\]
    If we also have  
    \[ \Sigma_{j=1}^{m_\eta}|u_j^{(\eta)}|_\#\leq N_\mu,\]
    this factorisation is the one we want. Otherwise, there is a child $\eta'$ of $\eta$ such that the factorisation $[w]_{\cR,\eta'}$  satisfies 
    \[ \mathcal{M} +\mu <\Sigma_{j=1}^{m_{\eta'}}|u_j^{(\eta')}|_\#\]
    and we may repeat the process beginning with $\eta$ instead of $\varepsilon$.

    Since the tree ${T}$ is finite, and by Lemma~\ref{lem:atv_2} at any leaf $l$ at most $\mu k|C||Q|$ letters of $w$ can be read, 
    continuing this process we can find a vertex $\nu$ such that
    \[ \mathcal{M} +\mu <\Sigma_{j=1}^{m_\nu}|u_j^{(\nu)}|_\#\leq N_\mu,\]
    as required. Finally, setting $i=i_{(w,\nu,\cR)}$, the left hand inequality implies
    that the set $U_i\in \cU$ corresponding to the \historyvec\
    of $w$ at $\nu$, with respect to $\cR$, is infinite. 
\end{proof}

\section{Applications of the Substitution Lemma }\label{sec:Applications}

As a first demonstration, we use Theorem~\ref{thm:Substitution}~\eqref{it:4.1.3} to prove the following. 

\begin{lemma}[{\cite[Lemma 3.3]{Seki}}]\label{lem:exampleXX}
Let  $m\in \N_+$ and $\Sigma=\{a_1,\dots, a_{2m+1}\}$ be an alphabet of size $2m+1$.
Then  $S_m=\{a_1^na_2^n\cdots a_{2m+1}^n\mid n\in\N_0\}$  is not $m$-\MCF.
\end{lemma}

\begin{proof}
Suppose $S_m\subseteq \Sigma^*$ is $m$-MCF. 
Setting $\mu=1$, let $M\in \N_+$ and sets $U_1,\dots, U_M$ be those given by  Theorem~\ref{thm:Substitution}, and 
$N_1\in \N_+$  the constant in Theorem~\ref{thm:Substitution}~\eqref{it:4.1.3}. 

For  $n>2N_1$ let $w[n]=a_1^{n}\cdots a_{2m+1}^{n}$ with all letters marked. 

If $w[n]=w_0u_1w_1$ with $1\leq |u_1|\leq N_1$ then either $u_1=a_j^sa_{j+1}^t$, with $s,t > 0$  and then $w_0w_1\notin S_m$, or
$u_1=a_j^s$ and again $w_0w_1=a_1^n\cdots a_{j-1}^na_j^{n-s}a_{j+1}^n\cdots a_{2m+1}^n\not\in S_m$. Hence $w[n]$ is not $(1,N_1)$-pumpable.

It follows that  there is an infinite \UpSet\ $U_i$ and a $U_i$-switchable factorisation of $w[n]=w_0u_1\cdots u_{m_i}w_{m_i}$ with $m_i\leq m$.
If, for some $j$ such that $1\le j<m_i$, $w_j$ contains a subword $a_{i-1}a_i^na_{i+1}$, for some $i$, then for all $(v_1,\ldots v_{m_i})\in U_i$ the word $w_0v_1\cdots v_{m_i}w_{m_i}$ is equal to $w[n]$, which is impossible since
$U_i$ is infinite. Similarly, if $w_0$ has a prefix $a_1^na_2$ or $w_{m_i}$ has a suffix $a_{2m}a_{2m+1}^n$, we obtain a contradiction. Therefore $w_0=a_1^{s_0}$ and $w_{m_i}=a_{2m+1}^{s_{m_i}}$, with $0\le s_0,s_{m_i}\le n$ and,  for $1\le j<m_i$, $w_j=a_{i_j}^{s_j}a_{i_j+1}^{t_j}$, for some
$i_j$ and $0\le s_j,t_j\le n$. This gives the bound
\[\sum_{j=0}^{m_i} |w_j|\le 2m_i n.\]
From Theorem~\ref{thm:Substitution}~\eqref{it:4.1.3} $\sum_{j=1}^{m_i}|u_j|\le N_1$ so we have $|w|\le 2m_in+N_1< (2m+1)n=|w|$, a contradiction.
\end{proof}

Note that the above result was first shown using 
 the weak pumping lemma of Seki {\em et al.} \cite{Seki}.
It is straightforward to prove that $S_m$ is $(m+1)$-\MCF\ (for example, using a $m+1$-restricted TSA where the \dendron\  built to accept $a_1^na_2^n\cdots a_{2m+1}^n$ is a single branch of length $n$, with the pointer pushing or moving up  the branch $m+1$ times reading odd-indexed letters, and down the branch $m$ times reading even-indexed letters).
This fact is observed in  \cite{Seki}, and yields \cite[Theorem 3.4]{Seki}: the class of $m$-\MCF\ languages is strictly contained in the class of  $(m+1)$-\MCF\ languages.

Next we use  Substitution Lemma  to prove  that the  language  $\{(a^mb^m)^n : m,n \geq 1 \}$ is not \MCF. Again, this 
result has already been shown using Seki's weak pumping lemma. 
The result shows that the Substitution Lemma is not already a consequence of having semi-linear Parikh image.

\begin{lemma}[{\cite[Lemma 6.14]{kallmeyer}}]\label{lem:exampleA}
Let 
 $L_1=\{(a^mb^m)^n \mid m,n \in\N_+ \}$. Then 
 \begin{itemize}\item the Parikh image of $L_1$ is semi-linear;
 \item $L_1$  is not multiple context-free.\end{itemize}
\end{lemma}

\begin{proof}
To see the first claim,  the Parikh image of $L_1$ is $\{(mn,mn)\mid m,n\in\N_+\}=\{(1,1)+k(1,1)\mid k\in \N_0\}$ which is  a linear subset of $\N_0^2$.

Let $\Sigma=\{a,b\}$. For $p\in\N_+$ a word $u\in \Sigma^*$ is said to be of \emph{type $p$} if it is of the form $cab^pad$ or $cba^pbd$, where $c,d\in \Sigma^*$.
If $u\in a^*b^*\cup b^*a^*$ then it's said to be of \emph{type $0$}. Note that words in $\Sigma^*$ may be of multiple types,
for instance  $ab^2a^3ba$  is of type $2, 3$ and $1$. A tuple  $(u_1,u_2,\dots,u_s)\in (\Sigma^*)^s$ is said to be
of  \emph{type $p$} if  it has  some coordinate $u_l$ of type $p\in \N_+$, and of \emph{type} $0$ otherwise.

 Suppose $w=(a^mb^m)^n$ for some  $m,n\in \N_+$. If  $w$  has a factorisation
\[w=w_0u_1w_1\dots w_{s-1}u_{s}w_{s}\] with $w_j,u_\ell\in \Sigma^*$ for all  $j\in[0,s], \ell\in[1,s]$,  then 
each factor $w_j,u_\ell$ has type $0$ or $m$.

Assume $L_1\subseteq \Sigma^*$ is $k$-MCF. 
 Set $\mu=1$ and let $M\in \N$ and sets $U_1,\dots, U_M$ be those given by  Theorem~\ref{thm:Substitution}, and 
let $N_1\in \N_+$ be the constant in Theorem~\ref{thm:Substitution}~\eqref{it:4.1.3}.
For each $m\in\N_+$, let \[w[m]=(\underline{a}a^{m-1}b^m)^{N_1+k+2}\]  where the underlined letters are marked (the first $a$ in each factor $a^m$) and the remaining letters are unmarked (that is, the marked letters are those in positions $2mj+1$, $j= 0, \dots, N_1+k+1$).
Then  $w[m]$ is of type $m$ and $|w[m]|_\#>N_1$, so by Theorem~\ref{thm:Substitution}~\eqref{it:4.1.3}  either (a) there is a factorisation $w[m]=w_0u_1w_1$ with $1\le |u_1|\le N_1$ and $w_0w_1\in L_1$; or (b) there exists $i\in[1,M]$ and $s\in [1,k]$ so that
$U_i$ is infinite and  $w[m]$ has a $U_i$-switchable factorisation \[w[m]=w_0u_1w_1\dots w_{s-1}u_{s}w_{s},\]
where $\sum_{l=1}^{s}|u_l|_\#\le N_1$, and each $w_j,u_l$ has type $0$ or $m$ (and no positive type $m'$ with $m\neq m'$).
If case (a) holds then, taking $m>N_1$ we have $|u_1|\le N_1$,  so $u_1$ has type $0$ and is not equal to $a^mb^m$ or $b^ma^m$,
so $w_0w_1$ cannot be in $L_1$, a contradiction. 

Hence, for $m>N_1$ case (b) must hold. In this event,  since there are $N_1+k+2$ marked letters,  we have $\sum_{j=0}^{s}|w_j|_\# \ge k+2\ge s+2$,
so there is some $w_j$ containing at least two marked letters. That is, some $w_j$ must
have type $m$ and, for each  $\mathbf{v}=(v_1,\dots, v_s)\in U_i$, as  $w_\mathbf{v}:=w_0v_1\cdots v_sw_s$ is in $L_1$ it
must also have type $m$. Now assume that $(v_1,\dots, v_s)\in U_i$ has  type $0$.
Then, as $w_{\mathbf{v}}$ is in $L_1$ and of type $m$  it follows that
$|v_\ell|\leq 2m$, since it has type $0$ but  cannot have a factor equal to $a^{m+1}$ or $b^{m+1}$.
%for each $\ell\in[1,s]$ otherwise $w_0v_1\cdots v_sw_s\not\in L_1$ since it is of type $m$ (due to $w_j$ being type $m$)
% and also contains a factor $a^{m+1}$ or $b^{m+1}$.
  It follows that if all tuples in $U_i$ have type $0$ then $U_i$ would be finite, a contradiction.
  Therefore $U_i$ contains some tuple of type $m$.

  We have now shown that, if $L_1$ is  $k$-MCF, then given $m>N_1$ there exists
  $i\in [1,M]$ such that $U_i$ is infinite, every tuple of $U_i$ is of type $0$
  or type $m$,
   no  tuple of $U_i$ has a positive type not equal to $m$, and $U_i$ 
  contains a tuple of type $m$. Choosing such an $i$ for integer $m>N_1$  we have an injective map from an infinite subset of $\N_+$ to $[1,M]$, a contradiction. 
\end{proof}

The next result follows by modifying the definition of \emph{type} in the  proof of Lemma~\ref{lem:exampleA} so that, for $p\in \N_+$,
a word is of \emph{type $p$} if it is of the form $cab^pad$  where $c,d\in \Sigma^*$, and of \emph{type $0$} if it is of the form
 $b^sab^t$, where $0\le s,t<p$, and considering words with all $a$ letters  marked. 
\begin{corollary}\label{cor:NotMCF}
    The language $L_2=\{(ab^m)^n \mid m,n \geq 1 \}$ is not multiple context-free.
\end{corollary}

For comparison, here is a sketch of  a proof of Lemma~\ref{lem:exampleA} using  the weak pumping lemma of Seki \emph{et al.} \cite{Seki}.
First we recall the weak pumping lemma. 
\begin{lemma}[{\cite[Lemma 3.2]{Seki}}]\label{lem:Seki}
For any infinite $k$-MCF language $L\subseteq \Sigma^*$, there exist some $u_j\in \Sigma^*, j\in[1,k+1]$ and $v_j,w_j,s_j\in \Sigma^*, j\in[1,k]$ which satisfy the following conditions.
\begin{\arabiclist}\item $\sum_{j=1}^k |v_js_j|>0$ and 
\item for any non-negative integer $i$, \[u_1v_1^iw_1s_1^iu_2v_2^iw_2s_2^iu_3\dots u_kv_k^iw_ks_k^iu_{k+1}\in L\]
\end{\arabiclist}
In other words, any infinite $k$-MCF language has at least one pumpable word in it (or a family that looks like a single word pumped).
\end{lemma}

\begin{proof}[Sketch of proof of Lemma~\ref{lem:Seki} as given in  {\cite[ Lemma 6.14]{kallmeyer}}]
Suppose for contradiction that $L_1$ is $k$-MCF for some $k\geq 1$. 
Let $L'$ be $L_1$ intersected with the regular language
\[\{a^{p_1}b^{q_1}\cdots a^{p_{k+1}}b^{q_{k+1}}\mid p_i,q_i \geq 1 \}.\]  
Then $L'$ is   $k$-MCF by Proposition~\ref{prop:closure_props}~\eqref{item:2}.

Each word in $L'$ has the form
$(a^mb^m)^{k+1}$ for some $m\in \N_+$. By  Lemma~\ref{lem:Seki}  there exists $m\in\N_+$ so that 
\[u_1v_1w_1s_1u_2v_2w_2s_2u_3\dots u_kv_kw_ks_ku_{k+1}=(a^mb^m)^{k+1}\] 
and for each $i\in \N$, 
\[w_i=u_1v_1^iw_1s_1^iu_2v_2^iw_2s_2^iu_3\dots u_kv_k^iw_ks_k^iu_{k+1}\in L'.\] 
 If any factor $v_j$ or $s_j$, $j \in [1, k]$ has both $a$ and $b$ letters, setting $i=2$ gives a word with 
more than $k+1$ factors of each letter, therefore each $v_i,s_i$ must contain only one type of letter.
This means that at most $2k$ of the factors $a^m,b^m$ contain $v_j,s_j$ and the remaining factors, of which there are at least $2$,
do not, so $w_i\not\in L'$, 
for $i\neq 1$, contrary to our assumption. Therefore  $L'$, and hence $L_1$, is not multiple context-free.\end{proof}

\subsection{The word problem of $F_2\times F_2$}\label{subsec:F2F2}

Recall that the word problem for a group $G$ with generating set $X$, denoted $\operatorname{WP}{(G,X)}$,  is the set of words in $\left(X\cup X^{-1}\right)^*$ which equal the identity element of the group (where $X^{-1}=\{x^{-1}\mid x\in X\}$ contains the inverse element for each $x\in X$).
 Since \MCF\ languages are closed under inverse homomorphism (Proposition~\ref{prop:closure_props}~\eqref{item:3}), having \MCF\ word problem is independent of choice of finite generating set.
Salvati's result shows that the word problem of $\Z^2$ is 2-\MCF\ (see also \cite{Nederhof,Caminati}).
  Ho \cite{Turbo} extended this to show that $\Z^n$ is $k$-\MCF\ for some $k>n$, and later Gebhardt \emph{et al.} \cite{ZnMCF} sharpened this by showing that in fact the word problem of $\Z^n$ is $n$-\MCF.
Since having a \MCF\ word problem is closed under finite index supergroup  \cite{Gilman2024},
   it follows that all finitely generated virtually  abelian groups have \MCF\ word problem, and since context-free languages are  \MCF, 
   virtually free groups have \MCF\ word problem by Muller-Schupp \cite{MS}.
 In addition, Kropholler and Spriano \cite{KS} show that having multiple context-free word problem is closed under taking  free products.
 Gilman,  Kropholler and Schleimer \cite{GKS} proved that nilpotent groups and some right-angled Artin groups including $F_2\times F_2$ (where $F_2$ is the free group of rank $2$)
 and $\GG(P_4)$, the right-angled Artin group with commutation graph $P_4$ (the tree with $2$ leaves and $2$ vertices of degree $2$), 
 do not have \MCF\ word problem by showing that the word problem intersected with some regular language over the generators yields a language with Parikh image that is not semi-linear (this contradicts that the word problem is \MCF\ since \MCF\ languages are closed under intersection with regular languages, and have semi-linear Parikh image). Thus having a \MCF\ word problem is not closed under direct product.
It is unknown if $F_2\times \Z$ has \MCF\ word problem.

Here we give a short alternative proof   that the word problem for $F_2\times F_2$ is not \MCF\ via the Substitution Lemma (Theorem~\ref{thm:Substitution}). 
 Later, Section~\ref{sec:RationalSubset}, we briefly consider decision problems for rational subsets of groups with multiple context-free word problem, which yields  
   even shorter proofs  that the word problems of $F_2\times F_2$ and $\GG(P_4)$ are not \MCF.

\begin{proposition}[{\cite[Theorem 25]{GKS}}]
\label{prop:F2xF2}
The word problem of $F_2\times F_2$ is not multiple context-free.
\end{proposition}
\begin{proof}
Let $\langle a,b,c,d\mid [a,c]=[b,c]=[a,d]=[b,d]=1\rangle$ be the standard presentation for $F_2\times F_2$ and fix the generating set $\Sigma=\{a,b,c,d, a^{-1},b^{-1},c^{-1}, d^{-1}\}$. For $w \in \Sigma^*$ we use the standard notations $w^*$ and $w^+$ to denote the sets $\{w^i\mid i\in \N_0\}$ and $\{w^i\mid i\in \N_+\}$ respectively.
Let $A=\langle a,b\mid-\rangle$ and $C=\langle c,d\mid-\rangle$ be subgroups isomorphic to $F_2$.
Let $L\subseteq \Sigma^*$ be the language of all words equal to the identity in $F_2\times F_2$. Assume for contradiction that $L$ is multiple context-free.

Define a regular language $T\subseteq \Sigma^*$ by the regular expression
\[T=\left((ca)^+(db)^+\right)^+(b^{-1})^+ \left(\left(d^{-1}a^{-1}\right)^+\left(c^{-1}b^{-1}\right)^+\right)^*\left(d^{-1}a^{-1}\right)^+ \left(c^{-1}\right)^+.\]

For $x_i,y_i,p_i,q_i\in\N_+$ let  $w\in T$ be a word of the  form
\begin{equation*}
\begin{split}
\left[(ca)^{x_1}(db)^{y_1}\cdots(ca)^{x_n}(db)^{y_n}\right] (b^{-1})^{p_1}
\left[\left(d^{-1}a^{-1}\right)^{q_1}\left(c^{-1}b^{-1}\right)^{p_2}\cdots \right.\\\left.
 \left(c^{-1}b^{-1}\right)^{p_t} \left(d^{-1}a^{-1}\right)^{q_{t}} \right]\left(c^{-1}\right)^{p_{t+1}}.
\end{split}
\end{equation*}

If $w\in T\cap L$, we have
\[a^{x_1}b^{y_1}\cdots a^{x_n}b^{y_n}
(b^{-1})^{p_1}\left(a^{-1}\right)^{q_1}\left(b^{-1}\right)^{p_2}\left(a^{-1}\right)^{q_2}\cdots \left(b^{-1}\right)^{p_t}  \left(a^{-1}\right)^{q_{t}}
=_A1,\]
which implies 
\begin{equation}\label{eqn1}
  y_n=p_1, \  x_n=q_1, y_{n-1}=p_2, \  x_{n-1}=q_2, \dots, y_1=p_t, x_1=q_t
  \textrm{ and } n=t.
\end{equation}

From the second factor we  have
\[c^{x_1}d^{y_1}\cdots c^{x_n}d^{y_n}
\left(d^{-1}\right)^{q_1}\left(c^{-1}\right)^{p_2}\left(d^{-1}\right)^{q_2}\cdots  \left(c^{-1}\right)^{p_t}\left(d^{-1}\right)^{q_t} \left(c^{-1}\right)^{p_{t+1}}=_C1,\] and, since $n=t$, we can write this as 
\[c^{x_1}d^{y_1}\cdots c^{x_n}d^{y_n}
\left(d^{-1}\right)^{q_1}\left(c^{-1}\right)^{p_2}\left(d^{-1}\right)^{q_2}\cdots  \left(c^{-1}\right)^{p_n}\left(d^{-1}\right)^{q_n} \left(c^{-1}\right)^{p_{n+1}}=_C1,\] 
which implies 
\begin{equation}\label{eqn2} y_n=q_1, x_n=p_2, y_{n-1}=q_2, x_{n-1}=p_3, \dots, y_1=q_n, x_1=p_{n+1}.\end{equation}

Putting Eqs.~\eqref{eqn1} and~\eqref{eqn2} together we get
\[p_1=y_n=q_1=x_n=p_2=y_{n-1}=q_2=x_{n-1}=\dots =p_n=y_1=q_n=x_1=p_{n+1}
\]
and so all powers in $w$ must be identical. Setting  $x_i=y_i=p_i=q_i=m$ we obtain 
\[w=((ca)^{m}(db)^{m})^n
(b^{-1})^{m}\left((d^{-1}a^{-1})^{m}(c^{-1}b^{-1})^{m}\right)^{n-1}(d^{-1}a^{-1})^{m}\left(c^{-1}\right)^{m}.\]
It follows that 
\[T\cap L\subseteq\{((ca)^{m}(db)^{m})^n
(b^{-1})^{m}\left((d^{-1}a^{-1})^{m}(c^{-1}b^{-1})^{m}\right)^{n-1}(d^{-1}a^{-1})^{m}\left(c^{-1}\right)^{m}
\mid m,n\in\N_+\}.\]

Conversely, since every word of the form 
\[((ca)^{m}(db)^{m})^n
(b^{-1})^{m}\left((d^{-1}a^{-1})^{m}(c^{-1}b^{-1})^{m}\right)^{n-1}(d^{-1}a^{-1})^{m}\left(c^{-1}\right)^{m}\]
 for $m,n\in \N_+$ evaluates to the identity in $F_2\times F_2$, and has the form of words in $T$, the above set inclusion is equality.
Applying the erasing homomorphism $\psi\colon \Sigma^*\to \{a,b\}^*$ defined by $a\mapsto a, b\mapsto b, x\mapsto \varepsilon$ for all $x\in\{c,c^{-1},d,d^{-1},a^{-1},b^{-1}\}$, we obtain 
 \[\psi\left(T\cap L\right)=\{(a^mb^m)^n\mid m,n\in \N_+\}\] which is not \MCF\ by Lemma~\ref{lem:exampleA}
  (itself a direct consequence of  Theorem~\ref{thm:Substitution}~\eqref{it:4.1.3}).
\end{proof}

\subsection{Discussion}
We finish this section with some further observations on the applicability and variations of Theorem~\ref{thm:Substitution}.

\subsubsection{Substitution lemma for context-free languages} 
Since $1$-MCF languages are exactly the context-free languages, 
one might think that Theorem~\ref{thm:Substitution} with $k=1$ yields a new way to prove a language is not context-free. However,
Theorem~\ref{thm:Substitution} with $k=1$ follows directly 
from Ogden's lemma for context-free languages as follows.

Recall Ogden's lemma.
\begin{lemma}[Ogden's lemma {\cite{Ogden}}]\label{lem:Ogden}
  Let $L$ be a context-free language. Then there exists $P\in \N_+$ such that, for all $w\in L$ with at least $P$ letters marked, there is a factorisation
  $w=aubvc$ of $w$ such that
  \begin{\arabiclist}
\item $|uv|_\#\ge 1$;
\item $|ubv|_\#\le P$ and
  \item for all $i\in \N$, $au^ibv^ic\in L$.
  \end{\arabiclist}
\end{lemma}

\begin{proof}[Proof that  Theorem~\ref{thm:Substitution}  for $1$-\MCF\ languages follows from {Lemma~\ref{lem:Ogden}}] Let  $L\subseteq \Sigma^*$ be an infinite context-free language and let $P$ be the constant guaranteed by {Lemma~\ref{lem:Ogden}}.
 Given $\mu\in \N_+$ define sets $U_{a_0,b_0,c_0,u,v}=\{(a_0u^ib_0,v^ic_0)\in(\Sigma^*)^2\mid i\in \N\}$ for all  $a_0,b_0,c_0,u,v\in \Sigma^*$ such that  $0<|uv|_\#$ and  $\mu\leq |a_0ub_0vc_0|_\#\leq \max\{\mu,P\}$, and note that this is a finite collection of infinite sets.

 Given $\lambda\in\N_+$ let $N_\lambda=\max\{\mu,P\}+\lambda$. If $w\in L$ and  $|w|_\#\geq N_\lambda$ then $|w|_\#\geq P$ so, from Ogden's lemma,
 $w=aubvc$ with $|uv|_\#>0$ and $|ubv|_\#\leq P$, so $|uv|_\#\leq P$.
 If $|uv|_\#<\mu$ let $a_0,b_0,c_0\in\Sigma^*$ be factors of $w$  such that  $a_0$ is a suffix of $a$, $b_0$ is a prefix of $b$ and $c_0$ is a prefix of $c$, and $|a_0ub_0vc_0|_\#=\mu$, and otherwise set $a_0=b_0=c_0=\varepsilon$; so $|a_0ub_0vc_0|_\#\leq \max\{\mu,P\}$.
 Thus  $(a_0ub_0,vc_0)\in U_{a_0,b_0,c_0,u,v}$,
  and $w$ has a factorisation
  $w=a_1a_0ub_0b_1vc_0c_1$ where $a_1$ is a prefix of $u$ and $b_1,c_1$ are suffixes of $b,c$ respectively, $|a_1b_1c_1|_\#\geq \lambda$
  since $|w|_\#=|a_0ub_0vc_0|_\#+|a_1b_1c_1|_\#$ and $|a_0ub_0vc_0|_\#\leq \max\{\mu,P\}$.
  {Lemma~\ref{lem:Ogden}} implies that    $a_1(a_0u^ib_0)b_1(v^ic_0)c_1\in L$ for all $i\in \N$. 
  This establishes item \eqref{it:4.1.2}.

Item~(3) follows from the previous paragraph after replacing $N_\lambda$ by $N_\mu=\max\{\mu,P\}$, yielding $\mu\leq |a_0ub_0vc_0|_\#\leq \max\{\mu,P\}=N_\mu$.

Item~(1) can be obtained by taking only the sets  $U_{a_0,b_0,c_0,u,v}$ where there exist $ a_1,b_1,c_1$ so that $a_1a_0ub_0b_1vc_0c_1\in L $.\end{proof}

\subsubsection{Non-application to a language of Lehner and Lindorfer}
In \cite{LL} Lehner and Lindorfer prove that for $a_1,\dots, a_k$ distinct letters, the
language \[L_k=\left\{a_1^{n_1}\cdots a_k^{n_k}\middle| n_1\geq \cdots \geq n_k \geq 0 \right\}\]  is $\lceil \frac{k}{2} \rceil$-MCF and
not $(\lceil \frac{k}{2} \rceil - 1)$-MCF. It is not difficult to show that $L_k$ is $\lceil \frac{k}{2} \rceil$-MCF, either by constructing a $\lceil \frac{k}{2} \rceil$-MCF grammar as in \cite{LL}, or directly by constructing a $\lceil \frac{k}{2} \rceil-$restricted TSA. The second part of the statement is shown in \cite{LL}  via simulating pumping on derivation trees of a grammar in Seki's normal form.
Here we observe that none of the statements of  Theorem~\ref{thm:Substitution} are sufficient to show $L_k$ is not $\lceil \frac{k}{2} \rceil$-MCF, in the case where $k=5$ and all letters are marked.

Given $\mu\in \N_+$ there exist the following infinite set $U_1=\{(a_1^i,\varepsilon)\mid i\in \N, i\ge \mu\}$,  
such that any word \[w=a_1^{n_1}a_2^{n_2}a_3^{n_3}a_4^{n_4} a_5^{n_5}\] of length at least $N=5(\mu-1)+1$
can be cut up into five factors $w_0u_1w_1u_2w_2$ with $(u_1,u_2)\in U_1$  as follows. As $w$ has length at least $5(\mu-1)+1$
it must have a prefix $a_1^\mu$, say $w=a_1^\mu w'$. Take  $u_1=a_1^\mu$, $w_0=w_1=u_1=\varepsilon$ and $w_2=w'$ to give a
factorisation $w=w_0u_1w_1u_2w_2$ with $(u_1,u_2)\in U_1$. This is a $U_1$-switchable factorisation as, for any $(v_1,v_2)\in U_1$, we have
$w''=w_0v_1w_1u_2w_2=a_1^iw'$, with $i\ge \mu$, so $w''\in L_5$.

Here the length constraints of the conclusion of Theorem~\ref{thm:Substitution}~\eqref{it:4.1.3} are met automatically, and choosing $N_\lambda\ge N +\lambda$, the same is true for Theorem~\ref{thm:Substitution}~\eqref{it:4.1.2}.

Indeed, if we restrict to the case where  all letters are marked,  Theorem~\ref{thm:Substitution} cannot even be used to show $L_5$ is not $1$-MCF! 
It is unclear to the authors whether there might be some way to use markings to show  $L_5$ is not $2$-MCF (or not $1$-MCF).

\subsubsection{Additional statement bounding $\sum_{j=0}^s|w_j|$ from above}
A version of Theorem~\ref{thm:Substitution}, where the sum of the marked lengths of the factors $w_j$ can be bounded above by a constant
can also be proved. For completeness we include this version below, even though we were unable to find a useful application for it. The proof 
follows lines similar to that of Theorem~\ref{thm:Substitution},
 but note that  the sets $U_i$ and associated parameters in the following are not necessarily the same as those in
 Theorem~\ref{thm:Substitution}. 

 For the proof, several variants of definitions in Section \ref{sec:prelims} are needed.  
    Let $\cA$, $Q$, $C$ and $D$ be as in the proof of Theorem~\ref{thm:Substitution}. We shall use a version of the \historyvec\ to
    record the states and labels of $\cA$ as the pointer moves away from the root and just before the pointer moves back down to the root.
    Let $w\in L$ and let $\cR=\t_1\cdots \t_r$ be a proper run of $\cA$ accepting $w$, with \dendron\ $T$. 
    Define the \emph{level-1-up-down array} of $w$ with respect to $\cR$ to be the $3\times s$ array
    \[l(w,\cR)=
      \begin{pmatrix}
      l_1 & l_2 & \ldots & l_s\\
      m_1 & m_2 & \ldots & m_s\\
      n_1 & n_2 & \ldots & n_s
    \end{pmatrix}
    \]
    with the following properties. Entries are all positive integers, $l_1< l_2< \cdots < l_s$, and 
    $l_i\le m_i$ for $i\in[1,s]$. 
    The integer $l_1$ is the smallest integer such that the transition $\t_{l_1}$ of $\cR$ contains a push command. This command pushes a vertex  $n_1$
    (incident to the root). The entries in the first and third row of column 1 of $l(w,\cR)$ are then $l_1$ and $n_1$, respectively.
    This is repeated for all other transitions which push vertices from the root. In detail,
    setting $i_1=1$, there is a maximal sequence of integers $l_{i_1}< l_{i_2}< \cdots < l_{i_d}$ such that,
    in the run $\cR$, the transition $\t_{l_{i_j}}$ contains a push command and the pointer is at the root before this transition is applied. Then
    the root has degree $d$ in $T$ and there are pairwise distinct integers $n_{i_j}$ such that $\t_{l_{i_j}}$ pushes the vertex $n_{i_j}$.
    The entries in the first and third row of columns $i_j$ are $l_{i_j}$ and $n_{i_j}$, respectively, for $1\le j\le d$.
    Every other column $i$ in the array corresponds to a transition $\t_{l_i}$ which contains an up command that is applied when the
    pointer is at the root and every such transition of $\cR$ corresponds to a unique column. If $\t_{l_i}$ contains an up command to a vertex
    $n_i$ then the third row of column $i$ equals $n_i$  (and by construction $n_i\in \{n_{i_1},\ldots, n_{i_d}\}$). Finally, for each $i$, the pointer remains above the root throughout transitions $\t_{l_i+1}, \ldots \t_{m_i}$ and is returned to the root by transition $\t_{m_i+1}$. Note that by definition the number of columns $s$ of the
    array is at most $kD$.

    Now define the \emph{level-1-factorisation} $[w]_{\cR,l=1}$ of $w$ by following Definition \ref{defnuf} but using the first two rows
    of $l(w,\cR)$ instead of $(l,m)_{\cR,\nu}$. To construct a \historyvec\ corresponding to $l(w,\cR)$ we combine the {\historyvec}s of each of the
    vertices adjacent to the root, keeping track of which child of the root the
    \historyvec\ is associated to, as follows. Fix $j$ such that $1\le j\le s$ and let $l(w,\cR)_j$ be the array produced by deleting all columns of
    $l(w,\cR)$ except those where the entry in the third row is $n_j$. Note that the first two rows of this array are equal to $(l,m)_{\cR,n_j}$
    (after rewriting the latter as a $2\times s$ array) and that the \historyvec\ $\harry_{w, \mathcal{R},n_j}$ of $w$ with respect to $n_j$ is determined by  $(l,m)_{\cR,n_j}$.
    Suppose that the columns of $l(w,\cR)$ that have $n_j$ in the third row are columns $i_1, \ldots ,i_{s_j}$, where $i_1< i_2< \cdots < i_{s_j}$. 
    The \emph{level-1-\historyvec} of $w$ with respect to  $\cR$ is the $3\times (2s)$ array $\harry_{w, \mathcal{R},l=1}$
    that has in the first two rows of its $(2i_{p}-1)$th and $2i_p$th columns the entries of the 
    $(2p-1)$th and $2p$th columns of $\harry_{w, \mathcal{R},n_j}$, respectively, and in its third row has $j$, for $1\le p\le s_i$ and $1\le j\le s$.
   
    Now, in the spirit of  Definition \ref{defU}, for  $s\in [1,kD]$, $\mathbf c\in C^{2s}$, $\mathbf{q}\in Q^{2s}$ and $\mathbf{n}\in \N^{2s}$ 
    we define the \emph{level-1-\UpSet}, $U_{(\mathbf c,\mathbf q, 1)}$,
    associated to the $3\times 2s$ array $\mathbf h=\left(\begin{smallmatrix}\mathbf{c}\\ \mathbf{q}\\ \mathbf{n}\end{smallmatrix}\right)$,
    to be the set of tuples $(u_1,\ldots ,u_s)$ such that there exists a proper run $\cR$ accepting a word $w\in L$ with level-1-factorisation
    $[w]_{\cR,l=1}=w_0u_1\cdots u_sw_s$ and $\harry_{w, \mathcal{R},l=1}=\mathbf h$. Again there are finitely many level-1-{\UpSet}s and
    we define $\cM_{l=1}$ by adjusting \eqref{eq:M} in the obvious way to capture the maximum sum of the lengths of $u_i$'s contained in the finite
    level-1-{\UpSet}s.

\begin{proposition}[Bounded Substitution Lemma]
\label{prop:Substitution_b} 
  Let  $L$ be an infinite  $k$-MCF language over an alphabet $\Sigma$.
  Given $\mu\in \N_+$ there exist integers $K,M, m_1,\dots, m_M\in\N_+$ 
  where $K\ge k$,  $m_i\in[1, K]$ and sets $U_1,\dots, U_M$ where  $U_i\subseteq (\Sigma^*)^{m_i}$  for $i\in[1,M]$ 
such that the following  holds.
\begin{\arabiclist}
\item\label{it:5.1.1} For each $i\in[1,M]$ and each  $(u_1,\dots, u_{m_i})\in U_i$, there exists a word $w\in L$  with $U_i$-switchable
  factorisation \[w=w_0u_1w_1\dots w_{m_i-1}u_{m_i}w_{m_i}.\]
\item\label{it:5.1.2} 
  There exists $N_\mu\in \N_+$ such that if $w\in L$ and at least $N_\mu$ letters of  $w$ are marked then either 
  \begin{alphlist}[(a)]
    \item $w$ is $(\mu,N_\mu)$-pumpable or 
    \item there exists $i\in [1,M]$, such that $U_i$ is infinite,  and $w$ has a $U_i$-switchable factorisation 
    \[w=w_0u_1w_1\dots w_{m_{i-1}}u_{m_i}w_{m_i}\]
  satisfying \[\Sigma_{j=0}^{m_i}|w_j|_\#\leq N_\mu.\]
  \end{alphlist}
\end{\arabiclist}
\end{proposition}

 \begin{proof}
    Observe that Lemma~\ref{1s} goes through on replacing $[w]_{\cR,\nu}$ and $[w']_{\cR',\nu'}$ with $[w]_{\cR,l=1}$ and  $[w']_{\cR',l=1}$, respectively, 
    and assuming that $\mathbf{h}_{w',\cR',l=1}=\mathbf{h}_{w,\cR,l=1}$. 
    To prove the proposition, first repeat the  proof of the first part of Theorem~\ref{thm:Substitution}, defining the
    sets $\{U_1,\ldots U_M\}$ based on level-1-\UpSet s instead of the original \UpSet s. Then the proof of part  \eqref{it:5.1.1} of the
    proposition follows from the proof of Theorem~\ref{thm:Substitution}~\eqref{it:4.1.1}.
    
    The first part of the proof of \eqref{it:5.1.2} above is exactly as the proof
    of Theorem~\ref{thm:Substitution}~\eqref{it:4.1.3}. For the second part set
    \[N_\mu=(D+1)(\mu k |C||Q|+\mu)+D\cM_{l=1}\] and as before assume there is no
    vertex of the \dendron\ at which the pointer  remains for more than
    $\mu|C||Q|$ consecutive transitions during the run $\cR$.
    Assume $w$ is a word in $L$ with at least $N_\mu$ letters marked and let
    $w=w_0u_1\cdots u_sw_s$ be the level-1-factorisation of $w$. From Lemma \ref{lem:atv_2} we have $\sum_{j=0}^s|w_j|\le \mu kD|Q|<N_\mu$ so
     $\sum_{j=0}^s|w_j|_\#<N_\mu$.
    If the level-1-\UpSet\, $U_i$ say, of $w$ with respect to $\cR$ is finite then $\sum_{j=1}^s|u_j|\le \cM_{l=1}$ so  
    $|w|_\#\le |w|<N_\mu$, a contradiction. Hence $U_i$ is infinite. 
 \end{proof}

\section{Rational subsets of groups with multiple context-free word problem}\label{sec:RationalSubset}

The following gives yet another way to obtain  Proposition~\ref{prop:F2xF2}.

Given a monoid $M$ generated by $\Sigma$, with generating map $\pi:\Sigma^*\rightarrow M$, a subset $R\subseteq M$ is said to be \emph{rational} if $R=\pi(L)$ for some regular language $L$ in $\Sigma^*$. (The subset $R$ is said to be \emph{recognisable} if $\pi^{-1}(R)$ is a regular language.) The \emph{rational subset membership problem} for $M$ is to decide, given a finite state automaton (FSA) over $\Sigma$ accepting a language $L$ and a word $w\in \Sigma^*$, whether or not $\pi(w)\in \pi(L)$.  
   The \emph{rational subset intersection (resp. equality) problem} for $M$ is to decide, given two FSAs over $\Sigma$ accepting   regular languages $L_1$ and $L_2$ in $\Sigma^*$,  whether or not $\pi(L_1)\cap\pi(L_2)=\emptyset$  (resp. $\pi(L_1)=\pi(L_2)$). 
The \emph{regular intersection decision problem for a language $L$}  is to decide, given one FSA over $\Sigma$ accepting a regular language $L_1$, whether $L\cap L_1$ is non-empty (see \cite[Section 3]{KambSilStei}).

It is straightforward to see that, in the setting of finitely generated groups, the rational subsets intersection problem is Turing equivalent to the rational subset membership problem: given two FSAs over $\Sigma$ accepting languages $L_1$ and $L_2$ in $\Sigma^*$, $\pi(L_1)\cap \pi(L_2)$ is non-empty if and only if $1\in \pi(L_1L_2^{-1})$ (which is a rational subset since $L_1L_2^{-1}$ is regular); conversely,  given an FSA over $\Sigma$ accepting $L$  and a word $x\in \Sigma^*$, the singleton set $\{\pi(x)\}$ is rational and $\{\pi(x)\}\cap \pi(L)$ is non-empty if and only if $\pi(x)\in \pi(L)$. 
Furthemore, by \cite[Theorem 3.1]{KambSilStei}, a finitely generated group $G$ has decidable rational subset membership problem if and only if $WP(G)$ has decidable regular intersection decision problem.

\begin{lemma}\label{mcftorsm} Let $G$ be a finitely generated group. If $WP(G)$ is multiple context-free then the rational subset membership problem is decidable in $G$ (and hence the rational subset intersection problem is decidable, and $WP(G)$ has decidable regular intersection decision problem).
\end{lemma}

\begin{proof} Assume $\Sigma$ is a finite inverse-closed generating set for $G$. Let $\mathcal B$ be a FSA accepting a language $L(\cB)$, where $R=\pi(L(\cB))$.  
As $\Sigma$ is inverse closed, if $w\in\Sigma^*$ represents $g\in G$,  then $w^{-1}\in \Sigma^*$ and $\pi(w^{-1})=g^{-1}$.

Without loss of generality, assume that $\cB$ has a unique accept state $f$. 
Attach a path labelled by $w^{-1}$ from $f$ to a new state $f_{\text{new}}$. Let   $\mathcal{B}_w$ be the resulting FSA, with single accept state $f_{\text{new}}$. Then $\mathcal{B}_w$ accepts exactly those words of the form $vw^{-1}$ where $\pi(v)\in R$, whence $\pi(L(\cB_w))=Rg^{-1}$.
It follows that  $z\in WP(G)\cap L(\cB_w)$ if and only if $z$ has the form  $vw^{-1}$, where $v\in L(\cB)$ and $\pi(v)=g$.
Hence $WP(G) \cap L(\cB_w)$ is non-empty if and only if $g\in R$.

Then the following algorithm decides membership in $R$. On input $w\in \Sigma^*$, 

\begin{\arabiclist}
\item build the FSA  $\cB_w$, as above

\item construct the TSA $\cA$ for $WP(G) \cap L(\cB_w)$.
This is possible since the TSA (or grammar) for the intersection of a multiple context
free language with a regular language is constructible \cite[Theorem 3.9]{Seki}

\item check whether or not $L(\cA)$ is empty, using the fact that  emptiness for multiple context free languages is decidable \cite{Kasami1} 
(see also \cite{Seki}).
\end{\arabiclist}

\end{proof}

We can then recover \cite[Theorems 24 and 25]{GKS} (see Section \ref{subsec:F2F2} above) as follows.
From \cite{lohrey}, a right-angled Artin group has decidable rational subset membership problem
 if and only if its defining graph has no full subgraph isomorphic to $C_4$ or $P_4$. 
   From \cite{DSS} a right-angled Artin group $G$ has a subgroup isomorphic to $\GG(P_4)$  if and only if its defining graph contains $P_4$ as a full subgraph and from \cite{Kambites_2009} (see also \cite{Blatherwick}) $G$ has a subgroup isomorphic to $\GG(C_4)\cong F_2\times F_2$  if and only if its defining graph contains $C_4$.

\section*{Acknowledgements}%+
We are grateful to Sarah Rees, Graham Campbell, Turbo Ho, Alex Levine and Mark Kambites for their invaluable input into this project. We thank the anonymous reviewer for their careful reading and helpful comments.
 The first author was supported 
 by The Leverhulme Trust, Research Project Grant RPG-2022-025.  The second author was supported by an Australian Research Council grant DP210100271 and by the London Mathematical Society Visiting Speakers to the UK--Scheme 2.
\bibliographystyle{plain}
\bibliography{refs_Swap_ijac}

\end{document}